%% file: main.tex
\documentclass[9pt]{article}
\usepackage{graphicx}
\usepackage{subfigure}
\usepackage{float}
\usepackage{amsfonts,amsmath,amssymb,amsbsy,amsthm,enumerate}

\usepackage{tikz}
\usetikzlibrary{positioning,chains,fit,shapes,calc}
\usetikzlibrary{calc,positioning,decorations.pathmorphing,decorations.pathreplacing}

\usetikzlibrary{decorations.markings}
\usetikzlibrary{decorations.pathmorphing}
\usetikzlibrary{arrows.meta}
\usetikzlibrary{positioning,chains,fit,shapes,calc}

\newcommand{\lct}{\mathrm{lct}}
\newcommand{\etal}{\textit{et~al.}}
\newcommand{\calV}{\mathcal{V}}
\newcommand{\tw}{\mathrm{tw}}
\newcommand{\Branch}{\mathrm{Br}}
\newcommand{\cC}{{\mathscr C\!}}

\usepackage{mathrsfs}

\newtheorem{theorem}{Theorem}[section]

\newtheorem{proposition}[theorem]{Proposition}
\newtheorem{lemma}[theorem]{Lemma}
\newtheorem{corollary}[theorem]{Corollary}

\definecolor{myblue}{RGB}{80,80,160}
\definecolor{mygreen}{RGB}{80,160,80}
\definecolor{myred}{RGB}{255,0,0}
\definecolor{mybrown}{RGB}{139,69,19}

\newtheoremstyle{case}{}{}{}{}{\bfseries}{:}{ }{}
\theoremstyle{case}
\newtheorem{case}{Case}

\numberwithin{subcase}{case}

\usepackage{marginnote}

\title{Transversals of Longest Cycles in Partial $k$-Trees and Chordal Graphs}
\author{Juan Guti\'errez\textsuperscript{1}\\
	{\footnotesize \textsuperscript{1}Departamento de Ciencia de la Computaci\'on}\vspace{-.2cm}\\
	{\footnotesize Universidad de Ingenier\'ia y Tecnolog\'ia (UTEC), Per\'u}
	\footnote{
		A previous version of this work was partially supported by FAPESP (Proc.~2015/08538-5).
  E-mail:
  jgutierreza@utec.edu.pe.}
}

\begin{document}

\maketitle

\begin{abstract}
  Let $\lct(G)$ be the minimum cardinality of a set of vertices that intersects
  every longest cycle of a 2-connected graph~$G$.
  We show
  that
  $\lct(G)\leq k-1$ if $G$ is a partial $k$-tree
  and that
  $\lct(G)\leq \max \{1, {\omega(G){-}3}\}$ if $G$ is chordal, where $\omega(G)$ is the cardinality of a maximum clique
  in~$G$.
  Those results imply that all longest cycles intersect in 2-connected 
  series parallel graphs and in 3-trees.
\end{abstract}

\section{Introduction}

\input{sectionIntroduction.tex}

\section{Paths, Cycles, and Attractors}\label{section:basic-concepts-pathscycles}

\input{sectionPathsCyclesAttractor.tex}

\section{Tree Decomposition and Branches}
\label{section:basic-concepts-treewidth}

\input{sectionTreeDecAndBranches.tex}

\section{Partial $k$-trees and Chordal Graphs}
\label{section:partialktrees-chordalgraphs}

\input{sectionPartialKTreesChordalGraphs}

\section{Our main technique}
\label{section:central-lemma}

\input{sectionOurMainTechnique}

\section{Result for Partial $k$-Trees}\label{section:lct-tw}

\input{sectionResultPartialKTrees.tex}

\section{Result for Chordal Graphs}\label{section:lct-chordal}

\input{sectionResultChordal}

\section{Concluding remarks} \label{section:conclusion}

\input{sectionConcludingRemarks}

\bibliography{bibliografia}

\end{document}

%% file: sectionIntroduction.tex
It is known that, in a 2-connected graph, 
every pair of longest cycles intersect each other.
A natural question is whether all longest
cycles have a vertex in common.
This has in general a negative answer, 
as the Petersen's graph shows.
Thus, it is interesting to look for a set
of vertices  that intersects every longest cycle of the graph.
Such a set is called a \emph{longest cycle transversal},
or just a \emph{transversal}.
The minimum cardinality of a transversal in a graph $G$ is denoted by~$\lct(G)$.
 It is interesting to search for good upper bounds for~$\lct(G)$.
Note that~$\lct(G)=1$ if and only if all longest cycles have a common vertex.

Consider a 2-connected graph $G$ with $n$ vertices.
Thomassen~\cite{Thomassen78} showed that $\lct(G) \leq \lceil n/3 \rceil$. 
This bound was improved by
Rautenbach and Sereni \cite{RautenbachS14}, who proved
that $\lct(G) \leq \lceil \frac{n}{3} - \frac{n^{2/3}}{36} \rceil$.
Jobson \etal{}~\cite{Jobson16}
showed that~$\lct(G)=1$ if $G$ is a dually chordal graph, a class of graphs
that includes doubly chordal, strongly chordal, and interval graphs.
They also mention that their proof can be applied to show that~$\lct(G)=1$ if $G$ is a 2-connected split graph.
Fernandes and the author \cite{Fernandes17} showed that 
$\lct(G) = 1$ if~$G$ is a 3-tree, and that $\lct(G) \leq 2$ if~$G$ is a partial 3-tree.
In this paper, we give results for $\lct(G)$ when $G$ is a partial $k$-tree and when $G$ is chordal.
A previous extended abstract containing these results was presented
at LATIN 2018 \cite{Gutierrez18}.

This paper is organized as follows.
In Section~\ref{section:basic-concepts-pathscycles}, we establish basic concepts on paths and cycles, which includes the
very important concept of attractor.
In Section~\ref{section:basic-concepts-treewidth},
we give definitions on tree decompositions
and branches.
In Section~\ref{section:partialktrees-chordalgraphs},
we define the classes of partial $k$-trees and chordal graphs.
In Section~\ref{section:central-lemma}, we state a
central lemma (Lemma \ref{lemma:core-lpt-lct}) that will be used in the next two sections.
In Section~\ref{section:lct-tw}, we show
that~$\lct(G) \leq k-1$ for every 2-connected 
partial~$k$-tree~$G$ (Theorem~\ref{theorem:lct<tw-1}) and, in Section~\ref{section:lct-chordal},
we show
that~$\lct(G)\leq \max \{1, {\omega(G){-}3}\}$ for every 2-connected chordal graph~$G$ 
(Theorem~\ref{theorem:lct-chordal}).
Finally, in Section~\ref{section:conclusion}, we present some concluding remarks.
In this paper, all graphs considered are simple and the notation used is standard \cite{BondyM08,Diestel10}.


%% file: sectionPathsCyclesAttractor.tex
Given two paths~$C'$ and~$C''$,
if~$C' \cup C''$ is a path or a cycle,
it is denoted by~$C' \cdot C''$.
For a pair of vertices~$\{a,b\}$ in a cycle~$C$, 
let~$C'$ and~$C''$ be the paths 
such that~$C = C' \cdot C''$ and~$V(C') \cap V(C'') = \{a,b\}$.  
We refer to these paths as the \emph{$ab$-parts}\index{parts of a cycle} of~$C$.
Moreover, we can extend this notation and define,
for a triple of vertices~$\{a,b,c\}$
in a cycle~$C$, the \emph{$abc$-parts} of~$C$;
and, when the context is clear, we denote by $C_{ab}$, $C_{bc}$, and $C_{ac}$ the corresponding~$abc$-parts of~$C$.

In what follows,
let~$G$ be a graph and let~$S \subseteq V(G)$. 
We say that~$S$ \emph{separates} vertices~$u$ and~$v$ if~$u$ and~$v$ are
in different components of~$G-S$.
Let~$X \subseteq V(G)$. We say that~$S$ separates~$X$
if~$S$ separates at least two vertices in~$X$.
We say that a path or cycle~$C'$ \emph{$k$-intersects}~$S$\index{$k$-intersects}
if~${|V(C') \cap S|=k}$.
Moreover, we also say that~$C'$~$k$-intersects~$S$ at~$V(C') \cap S$.
A path or cycle~$C'$ \emph{crosses}~$S$ if~$S$ separates~$V(C')$
in~$G$. Otherwise,~$S$ \emph{fences}~$C'$.\index{fence} 
If~$C'$ crosses~$S$ and~$k$-intersects~$S$, then we say 
that~$C'$~$k$-\emph{crosses}~$S$.\index{$k$-cross}
We also say that~$C'$~$k$-crosses~$S$ at~$V(C') \cap S$.
If~$C'$ is fenced by~$S$ and~$k$-intersects~$S$, then we say that~$C'$ 
is~$k$-\emph{fenced} by~$S$ (Figure \ref{fig:defs-fenced-crosses}).

\begin{figure}[h]
\centering
\subfigure[ ]{
\resizebox{.3\textwidth}{!}{\input{defs-fenced-crosses-1}}
}
\subfigure[ ]{
\resizebox{.3\textwidth}{!}{\input{defs-fenced-crosses-2}}
}
\caption{(a) A graph~$G$ with~$S=\{a,b,c,d\}$.
    (b) Consider, in $G$, paths~$P_1=v_1av_5$ and~$P_2=v_3cdbv_4$,
    and cycles~$C_1=v_1bv_2dv_1$ and~$C_2=v_3v_4cabv_3$. Then~$P_1$ and~$C_1$
    cross~$S$, and~$P_2$ and~$C_2$ are fenced by~$S$.
    Moreover,~$P_1$ 1-crosses~$S$,~$P_2$ is~$3$-fenced by~$S$,
    $C_1$ 2-crosses~$S$ and~$C_2$ is 3-fenced by~$S$.
    (Also note that path~$cd$ and cycle~$abda$ are fenced by~$S$.)
    Cycles~$C_2$ and~$v_1bcv_5av_1$ are~$S$-equivalent.
    }
\label{fig:defs-fenced-crosses}
\end{figure}
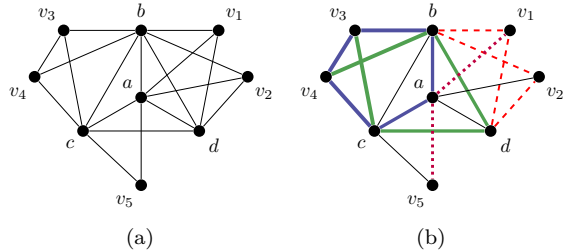

The length of a path or a cycle $C'$ in $G$
is the number of edges of $C'$ and it is denoted by $|C'|$.
A cycle in~$G$ is called a \emph{longest cycle}
if it has maximum length over all cycles in~$G$.
Two cycles are $S$-equivalent
if they intersect $S$ at the same set of vertices (Figure \ref{fig:defs-fenced-crosses}). Let~$C$ be a longest cycle in~$G$.
We say that~$C$ is an
\emph{attractor} for~$S$ if~$C$ is fenced by~$S$ and all
$S$-equivalent longest cycles are also fenced by~$S$.
We say that~$C$ is a~$k$-\emph{attractor} for $S$ if~$C$ $k$-intersects
$S$. In this case,
we also say that~$S$ \emph{has} a~$k$-attractor.

The next proposition is well-known, but, to our knowledge,
no simple proof of it has been written. We use it
several times through the text without making any reference
to it.

\begin{proposition}\label{prop:two-longest-cycles-intersect}
Let~$C$ and~$D$ be a pair of longest cycles in a 2-connected graph.
Then~${|V(C) \cap V(D)| \geq 2}$.
\end{proposition}
\begin{proof}
Suppose by contradiction that~${|V(C) \cap V(D)| \leq 1}$.
As~$G$ is 2-connected, there exist two disjoint paths~$R$
and~$S$, both of them with one extreme in~$C$, the other in~$D$,
and internally disjoint from both~$C$ and~$D$ \cite[Proposition~9.4]{BondyM08}.
Note that, when $|V(C) \cap V(D)| = 1$,
it can be the case that exactly one of~$\{R,S\}$ has zero length.
Let~$\{x_1\}= V(C) \cap V(R)$ and~$\{x_2\}=V(C) \cap V(S)$.
Let~$\{y_1\}= V(D) \cap V(R)$ and~$\{y_2\}=V(D) \cap V(S)$.
Let~$C'$ and~$C''$ be the two~$x_1x_2$-parts of~$C$.
Let~$D'$ and~$D''$ be the two~$y_1y_2$-parts of~$D$.
Then~$C' \cdot R \cdot D' \cdot S$ and~$C'' \cdot R \cdot D'' \cdot S$
are both cycles, one of them longer than~$|C|$, a contradiction.
\end{proof}

%% file: defs-fenced-crosses-1.tex
\begin{tikzpicture}
[scale=0.7, label distance=3pt, every node/.style={fill,circle,inner sep=0pt,minimum size=6pt}]
 
    \node at (0,0)[myblue,label=below left :$c$,fill=black, circle](c) {};
    \node at (3,0)[myblue,label=below right:$d$,fill=black, circle](d) {};
    \node at (1.5,2.5981)[myblue,label=above :$b$,fill=black, circle](b) {};
    \node at (3.5,2.5981)[myblue,label=above right:$v_1$,fill=black, circle](v1) {};
    \node at (4.25,1.4)[myblue,label=below right:$v_2$,fill=black, circle](v2) {};
    \node at (-0.5,2.5981)[myblue,label=above left:$v_3$,fill=black, circle](v3) {};

    \node at (1.5,0.866)[myblue,label=above   left:$a$,fill=black, circle](a) {};
    \node at (-1.25,1.4)[myblue,label=below left:$v_4$,fill=black, circle](v4) {};
    \node at (1.5,-1.4)[myblue,label=below left:$v_5$,fill=black, circle](v5) {};

    \draw  (a) -- (b);
    \draw  (b) -- (c);
    \draw  (a) -- (c);
    \draw (a) -- (d);
    \draw (b) -- (d);
    \draw (c) -- (d);
    
    \draw (b) -- (v1);
    \draw (d) -- (v1);
    \draw (b) -- (v2);
    \draw (d) -- (v2);
    
    \draw (b) -- (v3);
    \draw (c) -- (v4);
    \draw (v3) -- (v4);
    \draw (v3) -- (c);
    \draw (v4) -- (b);
    
    \draw (a) -- (v1);
    \draw (a) -- (v2);
    \draw (a) -- (v5);
    \draw (v5) -- (c);

\end{tikzpicture}

%% file: defs-fenced-crosses-2.tex
\begin{tikzpicture}
[scale=0.7, label distance=3pt, every node/.style={fill,circle,inner sep=0pt,minimum size=6pt}]
 
    \node at (0,0)[myblue,label=below left :$c$,fill=black, circle](c) {};
    \node at (3,0)[myblue,label=below right:$d$,fill=black, circle](d) {};
    \node at (1.5,2.5981)[myblue,label=above :$b$,fill=black, circle](b) {};
    \node at (3.5,2.5981)[myblue,label=above right:$v_1$,fill=black, circle](v1) {};
    \node at (4.25,1.4)[myblue,label=below right:$v_2$,fill=black, circle](v2) {};
    \node at (-0.5,2.5981)[myblue,label=above left:$v_3$,fill=black, circle](v3) {};

    \node at (1.5,0.866)[myblue,label=above   left:$a$,fill=black, circle](a) {};
    \node at (-1.25,1.4)[myblue,label=below left:$v_4$,fill=black, circle](v4) {};
    \node at (1.5,-1.4)[myblue,label=below left:$v_5$,fill=black, circle](v5) {};

    \draw  [color=myblue][line width=1.7pt] (a) -- (b);
    \draw  (b) -- (c);
    \draw  [color=myblue][line width=1.7pt] (a) -- (c);
    \draw (a) -- (d);
    \draw  [color=mygreen] [line width=1.7pt] (b) -- (d);
    \draw  [color=mygreen] [line width=1.7pt] (c) -- (d);
    
    \draw [dashed] [color=red] [line width=1.0pt] (b) -- (v1);
    \draw [dashed] [color=red] [line width=1.0pt] (d) -- (v1);
    \draw [dashed] [color=red] [line width=1.0pt] (b) -- (v2);
    \draw [dashed] [color=red] [line width=1.0pt] (d) -- (v2);
    
    \draw  (a) -- (v2);

    \draw [color=myblue][line width=2pt](b) -- (v3);
    \draw  [color=myblue][line width=2pt](c) -- (v4);
    \draw  [color=myblue][line width=2pt](v3) -- (v4);
    \draw  [color=mygreen] [line width=2pt] (v3) -- (c);
    \draw  [color=mygreen] [line width=2pt] (v4) -- (b);
    
    \draw  [dotted] [color=purple][line width=1.5pt](a) -- (v1);
    \draw   [dotted] [color=purple][line width=1.5pt] (a) -- (v5);
    \draw (v5) -- (c);

    \end{tikzpicture}

%% file: sectionTreeDecAndBranches.tex
A \emph{tree decomposition} \cite[p.~337]{Diestel10} of a graph~$G$
is a pair~$(T, \calV)$, consisting of a tree~$T$ and a collection
$\calV = \{ V_t: t \in V(T)\}$ of (different)
\emph{bags}~$V_t \subseteq V(G)$,
that satisfies the following three conditions:
\begin{itemize}
\item $\bigcup_{t \in V(T)} V_t = V(G);$
\item for every~$uv \in E(G)$, there exists
  a bag~$V_t$ such that~$u,v \in V_t;$
\item if~$v \in V(G)$ is in two different bags~$V_{t_1}$ and $V_{t_2}$,
  then~$v$ is also in any bag~$V_t$ such that~$t$ is on the
  path from~$t_1$ to~$t_2$ in~$T$.
\end{itemize}


The \emph{treewidth}~$tw(G)$ is the number
$\min \{ \max\{|V_t|-1: t\in V(T)\}: (T, \calV)$ is a
tree decomposition of $G  \}$.
We refer to the vertices of~$T$ as \emph{nodes}.


If~$G$ is a graph with treewidth~$k$,
then we say that~$(T, \mathcal{V})$ is a \emph{full tree decomposition}\index{full tree decomposition}
of~$G$ if~$|V_t|=k+1$ for every~$t\in V(T)$,
and~$|V_t \cap V_{t'}|=k$ for every~$tt' \in E(T)$
(Figure~\ref{fig:defs-tree-dec}).
\begin{proposition}[{\cite[Lemma 8]{Bodlaender98}}{\cite[Theorem 2.6]{Gross14}}]\label{prop:full-tree-dec}
Every graph has a full tree decomposition.
\end{proposition}

\begin{figure}[H]
\centering
\subfigure[ ]{
\resizebox{.3\textwidth}{!}{\input{defs-tree-dec-1}}
}
\subfigure[ ]{
\resizebox{.23\textwidth}{!}{\input{defs-tree-dec-2}}
}
\subfigure[ ]{
\resizebox{.3\textwidth}{!}{\input{defs-tree-dec-3}}
}
\caption{
(a) A graph~$G$ with treewidth two. 
(b) A tree decomposition of~$G$ that is not full.
 (c) A full tree decomposition of~$G$.
}
\label{fig:defs-tree-dec}
\end{figure}
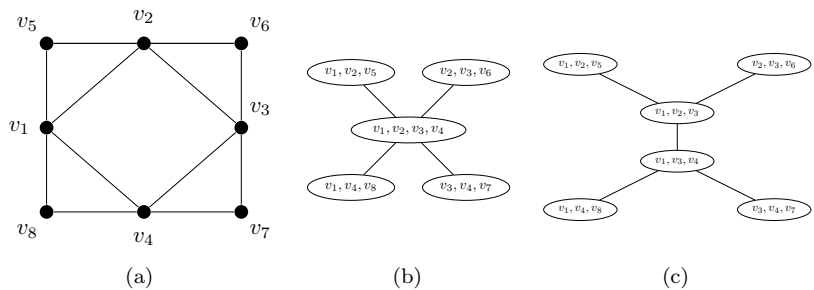


Let~$G$ be a graph and~$(T, \calV)$ be a tree decomposition of~$G$.
Given two different nodes~$t,t' \in V(T)$, we
denote by~$\Branch_t(t')$
the component of~$T-t$ where~$t'$ lies.
We say that such component is a \emph{branch}\index{branch} of~$T$ at~$t$,
and that the components
of~$T-t$ are the \emph{branches} of~$T$ at~$t$~\cite{Heinz13}.
Similarly, for a vertex~${v \notin V_t}$, it is denoted by 
$\Branch_t(v)$
the branch~$\Branch_t(t')$ of~$T$ at~$t$ 
such that~$v \in V_{t'}$.
In that case, we also say that~${v \in \Branch_t(t')}$.

Let~$t \in V(T)$.
Let~$C'$ be a path or cycle in~$G$ fenced by~$V_t$.
It is easy to see that, for every $u,v \in V(C') \setminus V_t$,
we have $\Branch_t(u)=\Branch_t(v)$.
Hence, when $V(C') \not \subseteq V_t$, we say that
$\Branch_t(C')=\Branch_t(v)$, where $v$ is an arbitrary vertex
of $V(C')\setminus V_t$. 
The next proposition relates the concepts of separation and branches.
\begin{proposition}[{\cite[Lemma~12.3.1]{Diestel10}}]\label{prop:core-sep-tt'}
Let~$G$ be a graph and~$(T, \calV)$ be a tree decomposition of~$G$.
  Let~$tt' \in E(T)$. 
  Let~$u,v \in V(G)$ be such that~${u \notin V_{t}}$ and~${v \notin V_{t'}}$.
  If~$u \in \Branch_t(t')$ and~$v \in \Branch_{t'}(t)$, then~$V_{t} \cap V_{t'}$
  separates~$u$ and~$v$.
\end{proposition}

%% file: defs-tree-dec-1.tex
\begin{tikzpicture}
[scale=0.5, label distance=3pt, every node/.style={fill,circle,inner sep=0pt,minimum size=6pt}]
    \node at (0,2.5981)[myblue,label=left :$v_1$,fill=black, circle](v1) {};
    \node at (0,0)[myblue,label=below left :$v_8$,fill=black, circle](v8) {};
    \node at (0,5.1962)[myblue,label=above left:$v_5$,fill=black, circle](v5) {};
    
    \node at (3,0)[myblue,label=below:$v_4$,fill=black, circle](v4) {};
    \node at (3,5.1962)[myblue,label=above:$v_2$,fill=black, circle](v2) {};
    
    \node at (6,0)[myblue,label=below right:$v_7$,fill=black, circle](v7) {};
    \node at (6,2.5981)[myblue,label=above right:$v_3$,fill=black, circle](v3) {};
    \node at (6,5.1962)[myblue,label=above right:$v_6$,fill=black, circle](v6) {};

    \draw (v1) -- (v2);
    \draw (v2) -- (v3);
    \draw (v3) -- (v4);
    \draw (v1) -- (v4);;

    \draw (v1) -- (v5);
    \draw (v2) -- (v5);
    \draw (v5) -- (v6);
    \draw (v2) -- (v6);
    \draw (v3) -- (v7);
    \draw (v3) -- (v6);
    \draw (v4) -- (v7);
    \draw (v4) -- (v8);
    \draw (v1) -- (v8);
        
\end{tikzpicture}

%% file: defs-tree-dec-2.tex
\begin{tikzpicture}
[scale=0.5, every node/.style={scale=0.7}]
 \node at (0,-2)[ellipse](t5) {$ $};
     \node at
    (2,2)[draw=black,ellipse](t1) {$v_1,v_2,v_3,v_4$};
    \node at (0,0)[draw=black,ellipse](t5) {$v_1,v_4,v_8$};
    \node at (4,0)[draw=black,ellipse](t4) {$v_3,v_4,v_7$};
    \node at (4,4)[draw=black,ellipse](t3) {$v_2,v_3,v_6$};
    \node at (0,4)[draw=black,ellipse](t2) {$v_1,v_2,v_5$};
    
    \draw (t1) -- (t2);
    \draw (t1) -- (t3);
    \draw (t1) -- (t4);
    \draw (t1) -- (t5);
\end{tikzpicture}

%% file: defs-tree-dec-3.tex

  \begin{tikzpicture}
[scale=0.5, every node/.style={scale=0.7}]
 \node at (0,-1.5)[ellipse](t5) {$ $};
     \node at
    (4,4)[draw=black,ellipse](t1) {$v_1,v_2,v_3$};
     \node at
    (4,2)[draw=black,ellipse](t2) {$v_1,v_3,v_4$};
    \node at (0,0)[draw=black,ellipse](t6) {$v_1,v_4,v_8$};
    \node at (8,0)[draw=black,ellipse](t5) {$v_3,v_4,v_7$};
    \node at (8,6)[draw=black,ellipse](t4) {$v_2,v_3,v_6$};
    \node at (0,6)[draw=black,ellipse](t3) {$v_1,v_2,v_5$};
    
    \draw (t1) -- (t2);
    \draw (t1) -- (t3);
    \draw (t2) -- (t6);
    \draw (t1) -- (t4);
    \draw (t2) -- (t5);
\end{tikzpicture}

%% file: sectionPartialKTreesChordalGraphs.tex
A~\emph{clique} in a graph is a set
of pairwise adjacent vertices.
A~$k$-\emph{clique} is a clique
of cardinality $k$.
The cardinality of a maximum clique in $G$
is denoted by $\omega(G)$.
A~$k$-\emph{tree}\index{$k$-tree} is defined recursively as follows.  
The complete graph on~$k$ vertices is a~$k$-tree.
Any graph obtained from a~$k$-tree by adding a new
vertex and making it adjacent to exactly 
all the vertices of an existing~$k$-clique is also a~$k$-tree.
A graph~$G$ is a \emph{partial} $k$-\emph{tree}\index{partial $k$-tree} if and only if
$G$ is the subgraph of a~$k$-tree. 
Partial~$k$-trees are closely related to the definition
of tree decomposition.
In fact, a graph~$G$ is a partial~$k$-tree if and only if~${\tw(G) \leq k}$ \cite[Theorem 35]{Bodlaender98}
(Figure \ref{fig:defs-k-tree}).

\begin{figure}[H]
\centering
\subfigure[ ]{
\resizebox{.35\textwidth}{!}{\input{defs-3-tree-1}}
}
\subfigure[ ]{
\resizebox{.35\textwidth}{!}{\input{defs-3-tree-2}}
}
\caption{
(a) A 3-tree~$G$. To construct~$G$, we begin with triangle~$abc$
and add the following sequence of vertices:~$v_1$-$v_2$-$v_3$-$v_4$-$v_5$-$v_6$-$v_7$.
(b) We can obtain a tree decomposition of~$G$ in the following way:
each time we add a new vertex, say~$v_i$, to an already existing triangle,
say~$xyz$, we also add a new node, with corresponding bag~$\{x,y,z,v_i\}$,
to the tree decomposition and we make it adjacent to
an already existing node whose corresponding bag contains~$x,y$ and~$z$. 
Moreover, the tree decomposition obtained by this procedure is a full
tree decomposition of~$G$.}
\label{fig:defs-k-tree}
\end{figure}
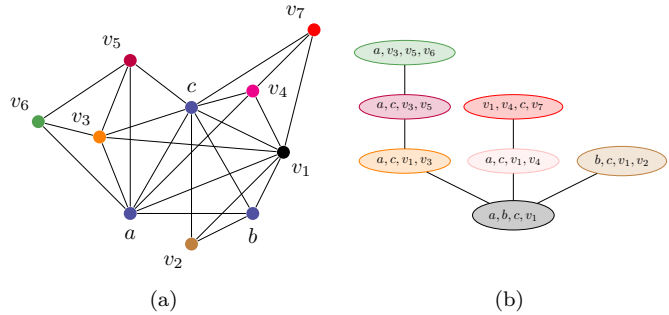

A graph is called \textit{chordal}\index{chordal graph} if every induced cycle has
length three.
A tree decomposition~$(T, \mathcal{V})$ of a graph~$G$
is called a \emph{clique tree}\index{clique tree} if~$\mathcal{V}$ is the set of all maximal cliques in~$G$ (Figure~\ref{fig:defs-chordal}).
\begin{proposition}[{\cite[Theorem 2, Theorem 3]{Gavril74}}]\label{prop:clique-tree}
  Every chordal graph has a 
  clique tree.  
\end{proposition}

\begin{figure}[h]
\centering
\subfigure[ ]{
\resizebox{.35\textwidth}{!}{\input{defs-chordal-1}}
}
\subfigure[ ]{
\resizebox{.45\textwidth}{!}{\input{defs-chordal-2}}
}
\caption{
(a) A chordal graph~$G$
with~$\omega(G)=5$. 
(b) A clique tree of~$G$. 
}
\label{fig:defs-chordal}
\end{figure}
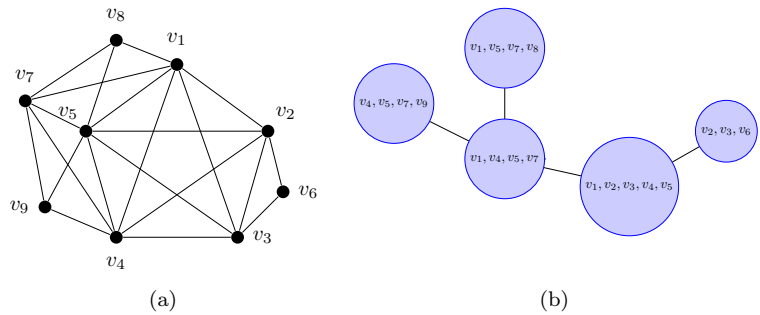

%% file: defs-3-tree-1.tex
\begin{tikzpicture}
[scale=0.5, label distance=3pt, every node/.style={fill,circle,inner sep=0pt,minimum size=6pt}]
\node at (0,0)[myblue,label=below :$a$,fill=myblue, circle](a) {}; 
\node at (4,0)[myblue,label=below:$b$,fill=myblue, circle](b) {};
\node at (2,3.464)[myblue,label=above:$c$,fill=myblue, circle](c) {};
\node at (5,2)[myblue,label=below   right:$v_1$,fill=black, circle](v1) {};
\node at (2,-1)[myblue,label=below left:$v_2$,fill=brown, circle](v2) {};
\node at (-1,2.5)[myblue,label=above   left:$v_3$,fill=orange, circle](v3) {};
\node at (4,4)[myblue,label= right:$v_4$,fill=magenta, circle](v4) {};
\node at (0,5)[black,label=above   left:$v_5$, fill=purple,circle](v5) {};
\node at (-3,3)[black,label=above   left:$v_6$, fill=mygreen,circle](v6) {};
\node at (6,6)[black,label=above   left:$v_7$, fill=red,circle](v7) {};

 \draw  (a) -- (b);
 \draw (b) -- (c);
 \draw (a) -- (c);
 
  \draw  (a) -- (v1);
   \draw  (b) -- (v1);
    \draw (c) -- (v1);
    
    \draw (b) -- (v2);
    \draw (c) -- (v2);
    \draw (v1) -- (v2);
    
      \draw  (a) -- (v3);
   \draw  (c) -- (v3);
    \draw  (v1) -- (v3);
    
         \draw  (a) -- (v4);
   \draw  (c) -- (v4);
    \draw  (v1) -- (v4);
    
            \draw  (a) -- (v5);
   \draw  (c) -- (v5);
    \draw  (v3) -- (v5);

         \draw  (a) -- (v6);
   \draw  (v5) -- (v6);
    \draw  (v3) -- (v6);    
    
         \draw  (c) -- (v7);
   \draw  (v1) -- (v7);
    \draw  (v4) -- (v7);    
    
  \end{tikzpicture}

%% file: defs-3-tree-2.tex
\begin{tikzpicture}
[scale=0.5, every node/.style={scale=0.7}]

\node[ ]  (V1) at (0,-2) {$ $};

\node[draw=black,fill=black!20,ellipse]  (V1) at (0,0) {$a,b,c,v_1$};
\node[draw=brown,fill=brown!20,ellipse]  (V2) at (4,2) {$b,c,v_1,v_2$};
\node[draw=orange,fill=orange!20,ellipse]  (V3) at (-4,2) {$a,c,v_1,v_3$};
\node[draw=pink,fill=pink!20,ellipse]  (V4) at (0,2) {$a,c,v_1,v_4$};
\node[draw=purple,fill=purple!20,ellipse]  (V5) at (-4,4) {$a,c,v_3,v_5$};
\node[draw=mygreen,fill=mygreen!20,ellipse]  (V6) at (-4,6) {$a,v_3,v_5,v_6$};
\node[draw=red,fill=red!20,ellipse]  (V7) at (0,4) {$v_1,v_4,c,v_7$};

\draw (V1) -- (V2);    
\draw (V1) -- (V3);
\draw (V1) -- (V4); 
\draw (V3) -- (V5);
\draw (V5) -- (V6);   
\draw (V4) -- (V7);
  \end{tikzpicture}

%% file: defs-chordal-1.tex
\begin{tikzpicture}
[scale=0.5, label distance=3pt, every node/.style={fill,circle,inner sep=0pt,minimum size=6pt}]

    \node at (5,-1.3)[myblue,label=above :$v_1$,fill=black, circle](v1) {};
    \node at (8,-3.5)[myblue,label=above right:$v_2$,fill=black, circle](v2) {};
    \node at (7,-7)[myblue,label= right:$v_3$,fill=black, circle](v3) {};
    \node at (3,-7)[myblue,label=below:$v_4$,fill=black, circle](v4) {};
    \node at (2,-3.5)[myblue,label=above left:$v_5$,fill=black, circle](v5) {};
    
    \node at (8.5,-5.5)[myblue,label= right:$v_6$,fill=black, circle](v6) {};
    \node at (0,-2.5)[myblue,label= above:$v_7$,fill=black, circle](v7) {};
    \node at (3,-0.5)[myblue,label=above:$v_8$,fill=black, circle](v8) {};
    \node at (0.65,-6)[myblue,label=left:$v_9$,fill=black, circle](v9) {};

    \draw  (v1) -- (v2);
    \draw  (v1) -- (v3);
    \draw  (v1) -- (v4);
    \draw  (v1) -- (v5);
    \draw  (v2) -- (v3);
    \draw  (v2) -- (v4);
    \draw  (v2) -- (v5);
    \draw  (v3) -- (v4);
    \draw  (v3) -- (v5);
    \draw  (v4) -- (v5);
    
    \draw  (v2) -- (v6);
    \draw  (v3) -- (v6);
    \draw  (v1) -- (v7);
    \draw  (v4) -- (v7);
    \draw  (v5) -- (v7);
    \draw  (v1) -- (v8);
    \draw  (v5) -- (v8);
    \draw  (v7) -- (v8);
    \draw  (v4) -- (v9);
    \draw  (v5) -- (v9);
    \draw  (v7) -- (v9);
    
\end{tikzpicture}

%% file: defs-chordal-2.tex
\begin{tikzpicture}
[scale=0.5, every node/.style={scale=0.7}]

\node[]  (Vv) at (0,-2) {$ $};

\node[draw=blue,fill=blue!20,circle]  (V1) at (0.5,1) {$v_1,v_2,v_3,v_4,v_5$};
\node[draw=blue,fill=blue!20,circle]  (V2) at (-4,2) {$v_1,v_4,v_5,v_7$};
\node[draw=blue,fill=blue!20,circle]  (V3) at (-4,6) {$v_1,v_5,v_7,v_8$};
\node[draw=blue,fill=blue!20,circle]  (V4) at (-8,4) {$v_4,v_5,v_7,v_9$};
\node[draw=blue,fill=blue!20,circle]  (V5) at (4,3) {$v_2,v_3,v_6$};

\draw (V2) -- (V2);    
\draw (V2) -- (V3);
\draw (V2) -- (V1);
\draw (V2) -- (V4);
\draw (V1) -- (V5);
\end{tikzpicture}

%% file: sectionOurMainTechnique.tex
In this section we introduce the technique for proving our results on
partial~$k$-trees and chordal graphs.
A similar technique and notation was introduced in~\cite{deRezende13}.
We begin by showing a new proof for the well-known Helly Property on trees (see~\cite{Helly23} and \cite{Horn72}).
Given a tree~$T$, a \emph{partial orientation}\index{partial orientation} of~$T$ is a digraph~$T'$
such that~$V(T')=V(T)$ and, if~$uv \in E(T')$, then
$uv \in E(T)$.
Note that not all edges
of~$T$ are present in~$T'$ as arcs.

\begin{lemma}[{\cite[Theorem 4.1]{Horn72}}]\label{lemma:helly-trees}
  Let~$T$ be a tree.
  Let~$\cC$ be a set of pairwise vertex-intersecting subtrees of~$T$.
  There exists a vertex $t \in V(T)$ such that every tree in~$\cC$ contains~$t$.
\end{lemma}
\begin{proof}
  We define a partial orientation~$T'$ of~$T$ as
  follows:~$tt' \in E(T')$ if and only if
  there exists a tree~$P \in \cC$, that does not contain~$t$,
  such that~$V(P)$ and~$t'$ are in the same component of~$T-t$.
  Suppose by contradiction that the lemma is false for~$T$.
  Then every node in~$T'$ has outdegree at least one.
  Let~$tt'$ be the last arc of a maximal directed path in~$T'$.
  As~$T$ is a tree,~$t't$ is also an arc in~$T'$, which implies 
  that there exist two trees~$P$ and~$Q$ in~$\cC$ such 
  that~$V(P)$ and~$t'$ are in the same component of~$T-t$,
  and~$V(Q)$ and~$t$ are in the same component of~$T-t'$.
  But then~$V(P) \cap V(Q) = \emptyset$, a contradiction
  (Figure~\ref{fig:lemma-helly-trees}).
\end{proof}

\begin{figure}[H]
        \centering
        {
\resizebox{.5\textwidth}{!}{\input{helly-trees}}
}
          \caption{The subtrees~$P$ and~$Q$ and the last arc $tt'$ of a maximal directed path in the proof of
          Lemma~\ref{lemma:helly-trees}.}
\label{fig:lemma-helly-trees}
\end{figure}
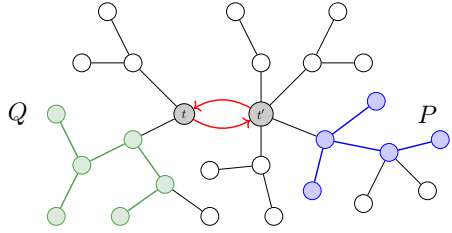

Our main technique for partial $k$-trees and chordal graphs is inspired
in the proof of Lemma \ref{lemma:helly-trees}, but adapted to the tree decomposition of the graph.
This is shown in Lemma~\ref{lemma:core-lpt-lct}. Before it, we show
a useful property.

 \begin{proposition}\label{prop:core-lpt-lct-Branch_t(P)=Branch_t(t')}
  Let~$(T, \calV)$ be a tree decomposition of a graph~$G$.
  If~$C$ is a path or cycle in~$G$ fenced by~$V_t$, for some~$t \in V(T)$,
  then either~$V(C) \subseteq V_t$ or there exists an edge~$tt'\in E(T)$ such that~$\Branch_t(C) = \Branch_t(t')$.
\end{proposition}
\begin{proof}
  If~$V(C) \subseteq V_t$, then there is nothing to prove. Otherwise,
  let~${u \in V(C) \setminus V_t}$.
  As~$u \notin V_t$, there exists a bag~$V_{t''}$ that contains~$u$.
  Let~$t'$ be the neighbor of~$t$ in~$T$ such that~$t'$ is in the
  path from~$t$ to~$t''$ in~$T$.
  Then~$\Branch_t(C) = \Branch_t(u) = \Branch_t(t'') = \Branch_t(t')$.
\end{proof}

The next lemma is crucial for our results on 
partial~$k$-trees and chordal graphs.

\begin{lemma}\label{lemma:core-lpt-lct}
  Let~$(T, \calV)$ be a tree decomposition of a graph~$G$.
  For every node~$t$, let~$\cC(t)$ be a set of cycles in~$G$ fenced by~$V_t$
  but not contained in~$G[V_t]$.
  If~$\cC(t) \neq \emptyset$ for every node~$t \in V(T)$,
  then there exists an edge~$tt' \in E(T)$ and two
  cycles~${C \in \cC(t)}$ and~$D \in \cC(t')$ 
  such that~$\Branch_t(C)=\Branch_t(t')$ and~$\Branch_{t'}(D)=\Branch_{t'}(t)$.
\end{lemma}
\begin{proof}
  We define a partial orientation~$T'$ of~$T$ as follows:
 ~$tt' \in E(T')$ if and only if~$tt' \in E(T)$ and
  there exists a cycle~$C \in \cC(t)$
  such that~$\Branch_t(C)=\Branch_t(t')$.
  For every~$t \in V(T)$,
  as~$\cC(t) \neq \emptyset$, there exists a
  cycle~$C$ fenced by~$V_t$ with~$V(C) \nsubseteq V_t$.
  Thus, by Proposition~\ref{prop:core-lpt-lct-Branch_t(P)=Branch_t(t')}, 
  there exists a neighbor~$t'$ of~$t$ in~$T$ such that~$\Branch_t(C)=\Branch_t(t')$.
  Hence every node in~$T'$ has outdegree at least one.
  Let~$tt'$ be the last arc of a maximal directed path in~$T'$.
  As~$T$ is a tree,~$t't$ is also an arc in~$T'$, which implies 
  that there exist two cycles~$C \in \cC(t)$ and~$D\in \cC(t')$
  such that~$\Branch_t(C)=\Branch_t(t')$ and~$\Branch_{t'}(D)=\Branch_{t'}(t)$.
\end{proof}

Immediate results are obtained for partial $k$-trees
and chordal graphs using Lemma \ref{lemma:core-lpt-lct} (see also \cite[Theorem 12.3.9]{Diestel10} and \cite[Proposition 2.6]{RautenbachS14}).
Recall that~$\omega(G)$ is the maximum cardinality of
a clique in $G$.
\begin{corollary} \label{corollary:lpt-lct-leq-tw+1}
Let~$G$ be a 2-connected graph.
Then~${\lct(G) \leq \tw(G)+1}$.
And, if $G$ is chordal, then~$\lct(G) \leq \omega(G)$.
\end{corollary}
\begin{proof}
It suffices to prove the first part, as the second part follows
directly by~\cite[Corollary 12.3.12]{Diestel10}.
Suppose by contradiction that ~$\lct(G) > \tw(G) +1 $
and let $(T,\calV)$ be a tree decomposition for $G$ of width $\tw(G)$.
Then, for every $t \in V(T)$, there exists a longest cycle that does
not intersect $V_t$. Thus, by Lemma \ref{lemma:core-lpt-lct}, there
exists an edge $tt' \in E(T)$ and two longest cycles $C$ and~$D$
such that~$C$ is fenced by $V_t$ and does not intersect $V_t$,
$D$ is fenced by $V_{t'}$ and does
not intersect~$V_{t'}$,
$\Branch_t(C)=\Branch_t(t')$ and $\Branch_{t'}(D)=\Branch_{t'}(t)$.
But then ${V(C) \cap V(D) = \emptyset}$, a contradiction.
\end{proof}

The main task in this paper is to improve the bounds
given by Corollary~\ref{corollary:lpt-lct-leq-tw+1}. 
So we have to find a longest cycle fenced by~$V_t$
that satisfies a particular property, which
will make our set~$\cC(t)$ nonempty for every~$t \in V(T)$,
to finally apply Lemma~\ref{lemma:core-lpt-lct}.
The main difficulty is that, when the bounds are diminished, the corresponding cycles can intersect several times the corresponding bag.

%% file: helly-trees.tex
\begin{tikzpicture}
[scale=0.5, every node/.style={scale=0.7}]
\node[draw=black,fill=black!20,circle]  (t) at (-1,0) {$t$};
\node[draw=black,fill=black!20,circle]  (tt) at (2,0) {$t'$};

\node[draw=black,circle,scale=1.5]  (v1) at (-3,2) {$ $};
\node[draw=black,circle,scale=1.5]  (v2) at (-5,2) {$ $};
\node[draw=black,circle,scale=1.5]  (v3) at (-4,4) {$ $};

\node[draw=mygreen,fill=mygreen!20,circle,scale=1.5]  (v4) at (-3,-1) {$ $};
\node[draw=mygreen,fill=mygreen!20,circle,scale=1.5]  (v5) at (-1.75,-2.75) {$ $};
\node[draw=mygreen,fill=mygreen!20,circle,scale=1.5]  (v6) at (-5,-2) {$ $};
\node[draw=mygreen,fill=mygreen!20,circle,scale=1.5]  (v9) at (-3.5,-4) {$ $};
\node[draw=mygreen,fill=mygreen!20,circle,scale=1.5]  (v7) at (-6,-4) {$ $};
\node[draw=mygreen,fill=mygreen!20,circle,scale=1.5]  (v8) at (-6,0) {$ $};
\node[draw=black,circle,scale=1.5]  (v10) at (0,-4) {$ $};

\node[draw=black,circle,scale=1.5]  (u1) at (6,2) {$ $};
\node[draw=black,circle,scale=1.5]  (u2) at (4,2) {$ $};
\node[draw=black,circle,scale=1.5]  (u3) at (5,4) {$ $};

\node[draw=black,circle,scale=1.5]  (u5) at (2,2) {$ $};
\node[draw=black,circle,scale=1.5]  (u6) at (1,4) {$ $};

\node[draw=blue,fill=blue!20,circle,scale=1.5]  (u7) at (9,-1) {$ $};
\node[draw=blue,fill=blue!20,circle,scale=1.5]  (u8) at (4.5,-1) {$ $};
\node[draw=blue,fill=blue!20,circle,scale=1.5]  (u9) at (7,-1.5) {$ $};
\node[draw=blue,fill=blue!20,circle,scale=1.5]  (u10) at (4,-3) {$ $};
\node[draw=black,circle,scale=1.5]  (u11) at (2,-2) {$ $};
\node[draw=blue,fill=blue!20,circle,scale=1.5]  (u12) at (6.5,0.5) {$ $};
\node[draw=black,circle,scale=1.5]  (u13) at (8.5,-3) {$ $};
\node[draw=black,circle,scale=1.5]  (u14) at (6,-3.5) {$ $};
\node[draw=black,circle,scale=1.5]  (u15) at (2.5,-4) {$ $};
\node[draw=black,circle,scale=1.5]  (u16) at (0,-2.2) {$ $};

\node[scale=2] at (-7.5,0) {$Q$};
\node[scale=2] at (8.5,0) {$P$};

\draw [->, red, thick] (t) to[bend right=30]  (tt); 
\draw [->, red, thick] (tt) to[bend right=30]  (t);

\draw (t) -- (v1);
\draw (v2) -- (v1);
\draw (v3) -- (v1);

\draw (t) -- (v4);
\draw [mygreen, thick](v6) -- (v4);
\draw[mygreen, thick] (v6) -- (v8);
\draw [mygreen, thick](v6) -- (v7);
\draw[mygreen, thick] (v4) -- (v5);
\draw [mygreen, thick](v9) -- (v5);
\draw (v10) -- (v5);

\draw (tt) -- (u5);
\draw (u5) -- (u6); 

\draw (tt) -- (u2);
\draw (u2) -- (u3); 
\draw (u2) -- (u1); 

\draw (tt) -- (u8);
\draw [blue, thick](u9) -- (u7); 
\draw[blue, thick] (u8) -- (u9);
\draw[blue, thick] (u8) -- (u12);
\draw [blue, thick](u8) -- (u10);
\draw (u14) -- (u9);
\draw (u13) -- (u9);
\draw (u11) -- (u15);
\draw (u11) -- (u16);

\draw (tt) -- (u11);
  \end{tikzpicture}

%% file: sectionResultPartialKTrees.tex
By
Corollary \ref{corollary:lpt-lct-leq-tw+1}, we have
that $\lct(G)\leq k+1$ when~$G$ is a 2-connected partial~$k$-tree.
In this section we improve this result and show that, in fact,~${\lct(G)\leq k-1}$
(Theorem \ref{theorem:lct<tw-1}).
We begin by showing a useful lemma.

\begin{lemma}\label{lemma:fenceds} \label{lct(G)>V_t-2-or-pulled}
  Let~$G$ be a 2-connected graph.
  Let~$(T, \calV)$ be a full tree decomposition of~$G$.   
  Let~$t \in V(T)$.
  If~${\lct(G)>|V_t|-2}$, then~$V_t$ has an {$\ell$-attractor} with~$\ell \leq 2$.
\end{lemma}
\begin{proof}
  As~${\lct(G)>|V_t|-2}$,
  for every subset of~$V_t$ with cardinality~${|V_t|-2}$, there exists a longest cycle that does not contain any vertex of it.
  If any of these cycles intersects~$V_t$ at most once, then there is
  an {$\ell$-attractor} for~$V_t$ with~$\ell \leq 1$ and we are done.
  Hence, every such cycle {2-intersects}~$V_t$.
  So, for every pair of vertices in~$V_t$, there exists a longest cycle that~$2$-intersects~$V_t$ at such pair.
  Suppose by contradiction that~$V_t$ has no {$\ell$-attractor}, with~$\ell \leq 2$.
  Then,
  for every pair of vertices in~$V_t$, there exists a longest cycle that~$2$-crosses~$V_t$ at~such pair.
  Observe that it cannot be the case that all such cycles contain an edge of~$V_t$.
  Hence, there exists a longest cycle~$C$ that 2-crosses~$V_t$, say at~$\{a,b\}$,
  such that~${ab \notin E(C)}$.
  
  Let~$C'$ and~$C''$ be the two~$ab$-parts of~$C$.
  As both~$C'$ and~$C''$ are fenced by~$V_t$ and are not contained in $V_t$, 
  by Proposition~\ref{prop:core-lpt-lct-Branch_t(P)=Branch_t(t')}, there exists two nodes~$t'$ and $t''$, neighbors
  of~$t$ in~$T$, such that~$\Branch_t(C')=\Branch_t(t')$ and~$\Branch_t(C'')=\Branch_t(t'')$,
  where possibly~$t'=t''$.
 As~$(T, \calV)$ is a full tree decomposition, we have~${|V_t \cap V_{t'}|=|V_t \cap V_{t''}|=|V_t|-1}$, so 
 ~$V_t \setminus V_{t'}$ consists on one vertex, say $x$.
  If~$V_t \cap V_{t'} \neq V_t \cap V_{t''}$, let
  $y$ be the vertex in~$V_t \setminus V_{t''}$.
  If~$V_t \cap V_{t'} = V_t \cap V_{t''}$, let
  $y$ be an arbitrary vertex in $V_t$ different from $x$.
  Let~$D$ be a longest cycle that 2-crosses~$V_t$ at~$\{x,y\}$ and let $D'$ and $D''$ be the two $xy$-parts
  of $D$. Note that $\Branch_t(D')$ and $\Branch_t(D'')$
  are different from both $\Branch_t(t')$ and $\Branch_t(t'')$.
  Then, 
  by Proposition \ref{prop:core-sep-tt'},~$C$ and $D$ intersect each other in at most one vertex, a contradiction to the fact that $G$ is 2-connected (Proposition \ref{prop:two-longest-cycles-intersect}).
  \end{proof}

Finally, we prove our main theorem.

\begin{theorem}\label{theorem:lct<tw-1}
For every 2-connected partial~$k$-tree~$G$, we have~$\lct(G) \leq k-1$.
\end{theorem}
\begin{proof}

  Let~$(T, \calV)$ be a full tree decomposition of~$G$.   
  For every~$t \in V(T)$, let~$\cC(t)$ be the set of longest cycles in~$G$
  such that, for every~$C \in \cC(t)$,~$C$ 
  is an~{$\ell$-attractor} for~$V_t$ with~$\ell \leq 2$.
  Suppose by contradiction that~$\lct(G) > k-1$.
  Then, as~$|V_t|=k+1$, by Lemma~\ref{lct(G)>V_t-2-or-pulled},~$\cC(t) \neq \emptyset$
  for every~$t \in V(T)$.
  Thus, by Lemma~\ref{lemma:core-lpt-lct}, there exists  
   an edge~$tt' \in E(T)$ and two longest
   cycles~$C$ and~$D$ in~$G$ such
   that~$\Branch_t(C)=\Branch_t(t')$,~$\Branch_{t'}(D)=\Branch_{t'}(t)$,
   ~$C$ is an {$\ell$-attractor} for~$V_t$ with~$\ell \leq 2$,
   and~$D$ is an~$\ell'$-attractor for~$V_{t'}$ with~$\ell' \leq 2$.  
  
  
  It is easy to see, by Proposition~\ref{prop:core-sep-tt'},
  that~$u \notin V(C)$. Analogously, we can conclude that~$w \notin V(D)$
  and therefore~$V(C) \cap V(D) \subseteq V_{t}\cap V_{t'}$. 
  Hence, as~$C$ and~$D$ are given by Lemma~\ref{lct(G)>V_t-2-or-pulled}
  and $G$ is 2-connected, we may assume that~${V(C) \cap V_t} = V(D) \cap V_t = \{a,b\}$.
  Let~$C'$ and~$C''$ be the two~$ab$-parts of~$C$.
  Let~$D'$ and~$D''$ be the two~$ab$-parts of~$D$.
  As~$|V(C) \cap V(D)|=2$, we can conclude that~$|C'|=|C''|=|D'|=|D''|$.
  We may assume that~$w \notin V(D')$.
  Hence,~$C' \cdot D'$ is a longest cycle 
  that 2-crosses~$V_t$ at~$\{a,b\}=V(C) \cap V_t$, a contradiction to the fact that~$C$ is
  a {2-attractor} for~$V_t$.
\end{proof}

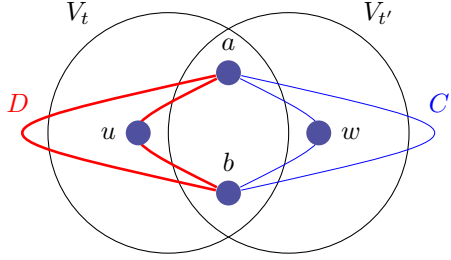
\begin{figure}[ht]
\centering

\begin{tikzpicture}[scale=0.8]
    \node at (1,2) {$V_t$};
    \node at (6,2) {$V_{t'}$};
    
    \node at (3.5,1)[myblue,label=above :$a$,fill=myblue, circle](a) {};
    
    
     \node at (3.5,-1)[myblue,label=above :$b$,fill=myblue, circle](b) {};
      
    \draw (2.5,0) ellipse (2 and 2);
    \draw (4.5,0) ellipse (2 and 2);
    
    \draw [color=blue] [line width=1.0pt, style =ultra thin] (a) .. controls (8, 0) .. (b);
    \draw [color=blue] [line width=1.0pt, style =ultra thin] (a) .. controls (5.5,0) .. (b);
    \node at (7,0.5) [color=blue]{$C$};
    
    \draw [color=red] [line width=1.0pt] (a) .. controls (-1, 0) .. (b);
    \draw [color=red][line width=1.0pt] (a) .. controls (1.5,0) .. (b);
    \node at (0,0.5) [color=red]{$D$};
    
    \node at (2,0)[myblue,label=left:$u$,fill=myblue, circle](u) {};
    
    \node at (5,0)[myblue,label= 
    right:$w$,fill=myblue, circle](w) {};
    \end{tikzpicture}

\caption{Situation in the proof of Theorem~\ref{theorem:lct<tw-1}.}
\label{fig:pairw-inters-1}
\end{figure}

The previous theorem implies the following result.

\begin{corollary}\label{corollary:series-parallel}
All longest cycles intersect in 2-connected partial 2-trees, also known as series-parallel graphs.
\end{corollary}

Also, we
have the following corollary due to results of
Fomin and Thilikos~\cite{FominT06}, and
Alon, Seymour, and Thomas~\cite{AlonST90}.

\begin{corollary}
  For every 2-connected planar graph~$G$ on~$n$ vertices, we have~${\lct(G) < 3.182 \sqrt{n}}$,
  and for every 2-connected~$K_r$-minor free graph~$G$, we have~${\lct(G) < r^{1.5}\sqrt{n}}$.
\end{corollary}

%% file: sectionResultChordal.tex
In this section, we prove that 
$\lct(G)\leq \max \{1, \omega(G)-3\}$ for every 2-connected chordal graph~$G$
(Theorem \ref{theorem:lct-chordal}).
Throughout this section, we denote by $L:=L(G)$ the length of a longest cycle in $G$. Recall that $\omega(G)$ is
the cardinality of a maximum clique in $G$.

\subsection{Proof of the main theorem}\label{subsection:lct-chordal-main-theorem}

The next lemma conceals the heart of the proof of Theorem~\ref{theorem:lct-chordal}.
The proof of that lemma is presented in Subsection~\ref{subsection:lct-chordal-main-lemma}.

\begin{lemma}\label{lemma:main-lemma}
  Let~$G$ be a 2-connected chordal graph
  such that~${\lct(G) > \max \{1, \omega(G)-3\}}$.
  Let~$k$ be an integer with~$k \ge 2$.
  For each maximal $k$-clique in~$G$, 
  there exists an~{$\ell$-attractor} with~$\ell \leq \min\{3, k-1\}$.
\end{lemma}

Using this lemma, we derive the main result of this section.

\begin{theorem}\label{theorem:lct-chordal}
For every 2-connected chordal graph~$G$,~$\lct(G)\leq \max \{1, 	\omega(G){-3}\}.$
\end{theorem}
\begin{proof}
  Let~$(T, \calV)$ be a clique tree of~$G$, which exists by Proposition \ref{prop:clique-tree}.
  If some clique in $\calV$ has cardinality one, then $|V(G)|=1$ and we are done.
  Thus, every clique in~$\calV$ has cardinality at least two.
  For every~${t \in V(T)}$, let~${\cC(t)}$ be the set of longest cycles in~$G$
  which are~{$\ell$-attractors} for~$V_t$, with~${\ell \leq \min\{3, |V_t|-1\}}$.
  Suppose by contradiction that~${\lct(G) > \max \{1, \omega(G)-3\}}$.
  Then, by Lemma~\ref{lemma:main-lemma}, ${\cC(t) \neq \emptyset}$ for every~${t \in V(T)}$ .
  Observe that, as~$V_{t}$ is a clique, any cycle in~$\cC(t)$
  has no edges in~$G[V_t]$; indeed, 
  otherwise, such cycle will contain all vertices of~$V_{t}$, 
  a contradiction to the fact that ${\ell \leq \min\{3, |V_t|-1\}}$.
  This implies that, for any~$t \in V(T)$,
  no cycle in~$\cC(t)$ is contained in~$G[V_t]$.
  Thus, by Lemma~\ref{lemma:core-lpt-lct}, there exists  
  an edge ${tt' \in E(T)}$ and two cycles
  ${C \in \cC(t)}$ and~${D \in \cC(t')}$
  such that
  ${\Branch_t(C)=\Branch_t(t')}$,~${\Branch_{t'}(D)=\Branch_{t'}(t)}$,
  $C$ is an~{$\ell$-attractor} for~$V_t$ with~${\ell \le \min\{3, |V_t|-1\}}$,
  and~$D$ is an~$\ell'$-attractor for~$V_{t'}$ with~${\ell' \le \min\{3, |V_{t'}|-1\}}$.
  
  
  Suppose for a moment that~${|V_{t}\cap V_{t'}| \leq \omega(G)-2}$.
  As~$\lct(G) > \max \{1, \omega(G)-3\}$, there exists a longest cycle
  that contains at most one vertex of~${V_{t}\cap V_{t'}}$. As such cycle must intersect both $C$
  and $D$ twice, this is a contradiction to Proposition \ref{prop:core-sep-tt'}.
  Hence~${|V_{t}\cap V_{t'}| \geq \omega(G)-1}$.
  Moreover, as both~$V_{t}$ and~$V_{t'}$ are maximal and different,
  we conclude that~${|V_{t}|=|V_{t'}|=\omega(G)}$.
  Let~${\{u\} = V_{t} \setminus V_{t'}}$ and~${\{w\} = V_{t'} \setminus V_t}$.
  It is easy to see that~$u \notin V(C)$ and~$w \notin V(D)$. 
  As~$G$ is 2-connected, $|V(C) \cap V(D)| \geq 2$, so we have the following cases.
  
 \setcounter{case}{0}  
  
  \begin{case}
  Both~$C$ and~$D$ 3-intersect~$V_t$ and~$V_{t'}$, respectively.
  
  Let~$V(C) \cap V_t=\{a,b,c\}$.
  Consider the case when~$V(D) \cap V_{t'}=\{a,b,c\}$.
  We may assume, without loss of generality,
  that~$w \notin C_{ab}$ and~$u \notin D_{ab}$.
  As~$(C-C_{ab}) \cdot D_{ab}$ and~$(D-D_{ab}) \cdot C_{ab}$
  are cycles,~$|C_{ab}|=|D_{ab}|$ and both are longest cycles.
  Hence,~$(C-C_{ab}) \cdot D_{ab}$ is a longest cycle that 3-crosses~$V_t$ 
  at~$V(C) \cap V_t$, a contradiction to the fact that~$C$ is an attractor for~$V_t$
  (Figure~\ref{fig:lct-chordal}(a)).
  Now suppose that~$V(D) \cap V_{t'}=\{b,c,d\}$, with~$d \neq a$.
  Then,~$C_{bc}\cdot C_{ca}\cdot ad \cdot D_{db}$
  and~$D_{bc}\cdot D_{cd}\cdot da \cdot C_{ab}$
  are cycles, one of them longer than~$L$, a contradiction.
  \end{case}
  
  \begin{case}
  Both~$C$ and~$D$ 2-intersect~$V_t$ and~$V_{t'}$, respectively.
  
  Let $\{a,b\}={V(C) \cap V_t=V(D) \cap V_{t'}}$.
  Let~$C'$ and~$C''$ be the two~$ab$-parts of~$C$.
  Let~$D'$ and~$D''$ be the two~$ab$-parts of~$D$.
  As~${C' \cdot D'}$, ${C' \cdot D''}$, ${C'' \cdot D'}$ 
  and~${C'' \cdot D''}$ are cycles,~${|C'|=|C''|=|D'|=|D''|=L/2}$.
  Without loss of generality, we may assume that~${u \notin V(D')}$.
  Hence,~${D' \cdot C'}$ is a longest cycle that 2-crosses~$V_t$ at 
 ~${V(C) \cap V_t}$, a contradiction to the fact that~$C$ is an attractor for~$V_t$
  (Figure~\ref{fig:lct-chordal}(b)).
  \end{case}
  
  \begin{case}
  $C$ {3-intersects}~$V_t$ and~$D$ {2-intersects}~$V_{t'}$.

  We may assume that~${V(C) \cap V_t=\{a,b,c\}}$
  and that~${V(D) \cap V_{t'}=\{a,b\}}$.
  Let~$D'$ and~$D''$ be the two~$ab$-parts of~$D$.
  Without loss of generality, we may assume that~${u \notin V(D')}$.
   As~$(C-C_{ab}) \cdot D'$ and~$(D-D') \cdot C_{ab}$
  are cycles,~$|C_{ab}|=|D'|$ and both are longest cycles.
  Hence,~${(C-C_{ab}) \cdot D'}$ is a longest cycle that 3-crosses~$V_t$ at~${V(C) \cap V_t}$, a contradiction to the fact that~$C$ is an attractor for~$V_t$
  (Figure~\ref{fig:lct-chordal}(c)).
  \end{case}
  This concludes the proof.	
\end{proof}

\begin{figure}[hbt]
\centering

\subfigure[ ]{
\begin{tikzpicture}[scale=0.75,rotate=-90, every node/.style={scale=0.9}]
    
    \node at (1,0) {$V_t$};
    \node at (8,0) {$V_{t'}$};

    \node at (4.5,1)[myblue,label=right:$a$,fill=myblue, circle](a) {};
    
    \node at (4.5,0)[myblue,label=above:$b$,fill=myblue, circle](b) {};
    
     \node at (4.5,-1)[myblue,label=left:$c$,fill=myblue, circle](c) {};
      
    \draw (3.5,0) ellipse (2 and 2);
    \draw (5.5,0) ellipse (2 and 2);
    
    \draw [color=blue] [line width=1.0pt, style =ultra thin] (a) .. controls (9, 2) .. (b);
    \draw [color=blue] [line width=1.0pt, style =ultra thin] (b) .. controls (9, -2) .. (c);
    \draw [color=blue] [line width=1.0pt, style =ultra thin] (a) .. controls (7.15,0) .. (c);
    \node at (7.5,2) [color=blue]{$C$};
    
    \draw [color=red] [line width=1.0pt] (a) .. controls (0, 2) .. (b);
    \draw [color=red] [line width=1.0pt] (b) .. controls (0, -2) .. (c);
    \draw [color=red][line width=1.0pt] (a) .. controls (1.85,0) .. (c);
    \node at (1.5,2) [color=red]{$D$};
    
    \node at (2.5,0)[myblue,label=above left:$u$,fill=myblue, circle](u) {};
    
    \node at (6.5,0)[myblue,label=below:$w$,fill=myblue, circle](w) {};
    \path	
    
    (a) edge (b)
    (a) edge [bend left] (c)
    (b) edge (c)
    ;

    \end{tikzpicture}

}\subfigure[ ]{
  \begin{tikzpicture}[scale=0.75,rotate=-90, every node/.style={scale=0.9}]

    \node at (1,2) {$V_t$};
    \node at (6,2) {$V_{t'}$};
    
    \node at (3.5,1)[myblue,label=right :$a$,fill=myblue, circle](a) {};
    
     \node at (3.5,-1)[myblue,label=left :$b$,fill=myblue, circle](b) {};
      
    \draw (2.5,0) ellipse (2 and 2);
    \draw (4.5,0) ellipse (2 and 2);
    
    \draw [color=blue] [line width=1.0pt, style =ultra thin] (a) .. controls (8, 0) .. (b);
    \draw [color=blue] [line width=1.0pt, style =ultra thin] (a) .. controls (5.5,0) .. (b);
    \node at (7,0.5) [color=blue]{$C$};
    
    \draw [color=red] [line width=1.0pt] (a) .. controls (-1, 0) .. (b);
    \draw [color=red][line width=1.0pt] (a) .. controls (1.5,0) .. (b);
    \node at (0,0.5) [color=red]{$D$};
    
    \node at (2,0)[myblue,label= above:$u$,fill=myblue, circle](u) {};
    \node at (5,0)[myblue,label= 
    below:$w$,fill=myblue, circle](w) {};
    
    \path	
    
    (a) edge (b)
    ;
    
    \end{tikzpicture}
}
\subfigure[ ]{
 \begin{tikzpicture}[scale=0.75,rotate=-90, every node/.style={scale=0.9}]
    \node at (1,0) {$V_t$};
    \node at (8,0) {$V_{t'}$};
    \node at (4.5,1)[myblue,label=right :$a$,fill=myblue, circle](a) {};
    \node at (4.5,0)[myblue,label=below :$b$,fill=myblue, circle](b) {};
     \node at (4.5,-1)[myblue,label=left :$c$,fill=myblue, circle](c) {};
      
    \draw (3.5,0) ellipse (2 and 2);
    \draw (5.5,0) ellipse (2 and 2);
    
    \draw [color=blue] [line width=1.0pt, style =ultra thin] (a) .. controls (9, 2) .. (b);
    \draw [color=blue] [line width=1.0pt, style =ultra thin] (b) .. controls (9, -2) .. (c);
    \draw [color=blue] [line width=1.0pt, style =ultra thin] (a) .. controls (7.15,0) .. (c);
    \node at (7.5,2) [color=blue]{$C$};
    
    \draw [color=red] [line width=1.0pt] (a) .. controls (0, 2) .. (b);
    \draw [color=red][line width=1.0pt] (a) .. controls (1.85,-0.15) .. (b);
    \node at (1.5,2) [color=red]{$D$};
    
    \node at (2.5,0)[myblue,label=above left:$u$,fill=myblue, circle](u) {};
    
    \node at (6.5,0)[myblue,label=below:$w$,fill=myblue, circle](w) {};
    
    \path	
    
    (a) edge (b)
    (a) edge [bend right] (c)
    (b) edge (c)
    ;
    
    \end{tikzpicture}

}
\caption{Cases in the proof of Theorem~\ref{theorem:lct-chordal}.}
\label{fig:lct-chordal}
\end{figure}
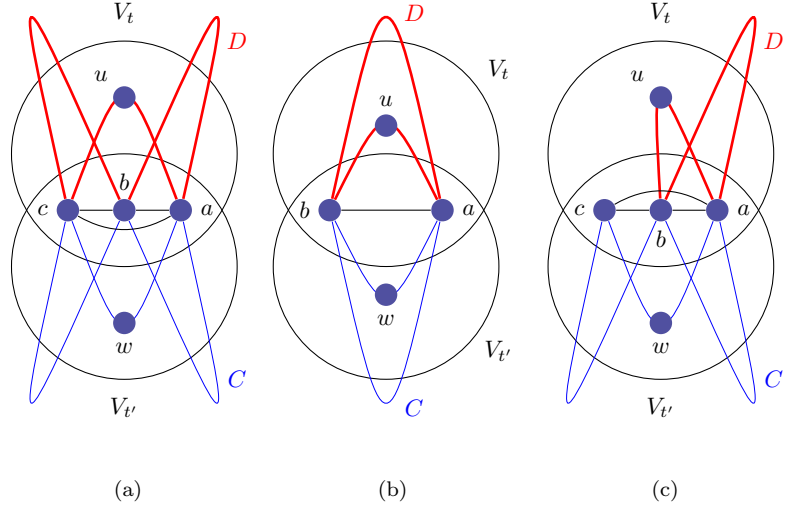

Note that~$k$-trees are chordal~\cite[Theorem 4.1]{Golumbic04}
and their maximum cliques have cardinality at most $k+1$.
Also, every planar graph is $K_5$-free.
Hence, we have the following corollary.

\begin{corollary}\label{corollary:lct-corol-chordal-ktree}
  If~$G$ is a~$k$-tree, with~$k>2$, then~$\lct(G)\leq k-2$.
  Moreover, all longest cycles intersect in 2-trees, 3-trees, and in 2-connected chordal planar graphs.
\end{corollary}

\subsection{Proof of the main lemma}\label{subsection:lct-chordal-main-lemma}

We next show the proof of Lemma~\ref{lemma:main-lemma}.
Before that, we present new useful definitions.
If~$C'$ and~$D'$ are paths fenced by a set of vertices~$K$ in a graph~$G$,
we write~$C' \sim_{K} D'$ if there exist vertices~${u \in V(C')}$
and~$v \in V(D')$ such that~$u$ and~$v$ are in the same component of~$G-K$.
Otherwise, we write~$C' \nsim_{K} D'$.
If the context is clear, we write~$C' \sim D'$ and~$C' \nsim D'$.
Given a cycle~$C$ that 3-crosses~$K$ at~$\{a,b,c\}$, we say 
that~$a$ \emph{breaks}~$C$\index{break} if~$C_{ab} \nsim_{K} C_{ac}$.
If the context is clear, we also say that~$a$ is a~$C$-\emph{breaking vertex}\index{breaking vertex} 
or that~$a$ is a \emph{breaking vertex}
(Figure~\ref{fig:defs-simeq}).
Recall that two paths or cycles $C'$ and
$C''$ are $K$-equivalent if
$V(C') \cap K = V(C'') \cap K$.

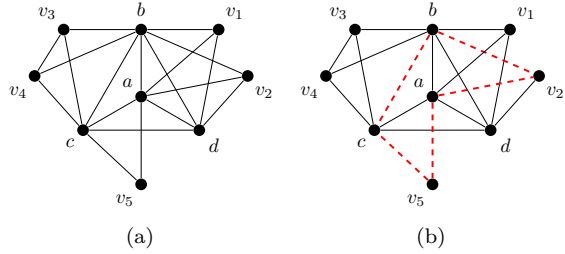
\begin{figure}[h]
\centering
\subfigure[ ]{
\resizebox{.3\textwidth}{!}{\input{defs-fenced-crosses-1}}
}
\subfigure[ ]{
\resizebox{.3\textwidth}{!}{\input{defs-fenced-crosses-3}}
}
\caption{(a) A graph with~$K=\{a,b,c,d\}$.
    (b) Consider 
    cycle~$C=av_2bcv_5a$.
    Then
    $C_{ab} \nsim_K C_{ac}$,  $C_{ab} \nsim_K C_{bc}$ and  $C_{bc} \nsim_K C_{ac}$.
    Hence, $a$ breaks~$C$.
    }
\label{fig:defs-simeq}
\end{figure}

\begin{proof}[Proof of Lemma~\ref{lemma:main-lemma}]
Let~$K$ be a maximal clique in~$G$ and~$k \geq 2$ be its cardinality.
If there is a longest cycle that intersects~$K$ at most once, then we are done.
  Indeed, such cycle would be an {$\ell$-attractor} for~$K$ with~$\ell\in \{0,1\}$.
  Hence, we may assume 
  that every longest cycle intersects~$K$ at least twice. Note that this implies, as $\lct(G)>1$,
  that $k \geq 3$.
  Let~${(T, \calV)}$ be a clique tree of~$G$. We have the next two cases.

   \setcounter{case}{0}

  \begin{case} There exists a longest cycle that
  2-intersects $K$.
  
  Let $\{a,b\} \subseteq K$ be such that
  there exists a longest cycle that 2-intersect $K$ at $\{a,b\}$. If all cycles that 2-intersect $K$ at $\{a,b\}$ are fenced by $K$, then we have an $\ell$-attractor with $\ell \leq 2$ and we are done. Hence, there
  exists a longest cycle $C$ that 2-crosses~$K$ at~$\{a,b\}$.
As~$K$ is a maximal clique in~$G$, there exists
${t \in V(T)}$ such that~${V_t=K}$. 
It is straightforward to see that, as $k \geq 3$, no $ab$-part of $C$ is an edge.
Let $C'$ and $C''$ be the two $ab$-parts of $C$.
By Proposition \ref{prop:core-lpt-lct-Branch_t(P)=Branch_t(t')},
there exists
two edges $tt_1,tt_2 \in E(G)$
such that~${\Branch_{t}(C')=\Branch_{t}(t')}$
and~${\Branch_{t}(C'')=\Branch_{t}(t'')}$.
As $V_t$, $V_{t'}$ and $V_{t''}$
are maximal cliques in $G$, there exist two vertices
$c \in V_t \setminus V_{t'}$ and
$d \in V_t \setminus V_{t''}$.
Note that $c,d \notin \{a,b\}$.
Note also that, although~$C'$ and~$C''$ are in different components of $G-K$, it can be the case that $t'=t''$, implying that $c=d$.
Let us assume that $c \neq d$ (so $k\geq 4$), as the proof when $c=d$ is very similar.
As ${\lct(G) > \max \{1, \omega(G)-3\}} \geq k-3$, there exists
a longest cycle $D$ 
that does not contain any vertex of $V_t \setminus \{b,c,d\}$.
If $D$ and any cycle equivalent to $D$, is fenced by $K$, then we are done.
Indeed, as $D$ will be a $\ell$-attractor with $\ell \leq  3 \leq \min\{3,k-1\}$.
So, we may assume, without loss of generality that
$D$ crosses $K$.

Suppose for a moment that $D$ 2-intersects $K$ at $\{x,y\}$
(note that it can be the case that $\{x,y\} \cap \{b\} \neq \emptyset$).
Let $D'$ and $D''$ be the two $xy$-parts of $D$.
As both $C$ and $D$ cross $K$, we may assume, without
loss of generality, that $C'$ is internally disjoint from $D'$ and that $C''$ is internally disjoint from $D''$.
But then, $C' \cdot bx \cdot D' \cdot ya$
and $C'' \cdot bx \cdot D'' \cdot ya$ are both cycles,
one of them longer than~$L$, a contradiction (Figure \ref{fig:lemma-lct-2-crosses-nuevo-total}$(a)$).
Hence, we may assume that
$D$ 3-intersects $K$.
If $V(D) \cap K \cap \{a,b\} = \emptyset$ then the proof is very similar to the previous case.
So, let us assume that $D$ intersects
$K$
at $\{b,c,d\}$.
Recall that 
$c \in V_t \setminus V_{t'}$ and
$d \in V_t \setminus V_{t''}$. Thus,
by Proposition \ref{prop:core-sep-tt'}, 
we have
${C' \nsim_{K} D_{bc}}$ and
${C' \nsim_{K} D_{cd}}$.
Analogously,
${C'' \nsim_{K} D_{bd}}$.
Hence, $C' \cdot D_{bc} \cdot D_{cd} \cdot da$
and $C'' \cdot D_{bd} \cdot da$
are both cycles, one of them longer than $L$, a contradiction (Figure \ref{fig:lemma-lct-2-crosses-nuevo-total}$(b)$).

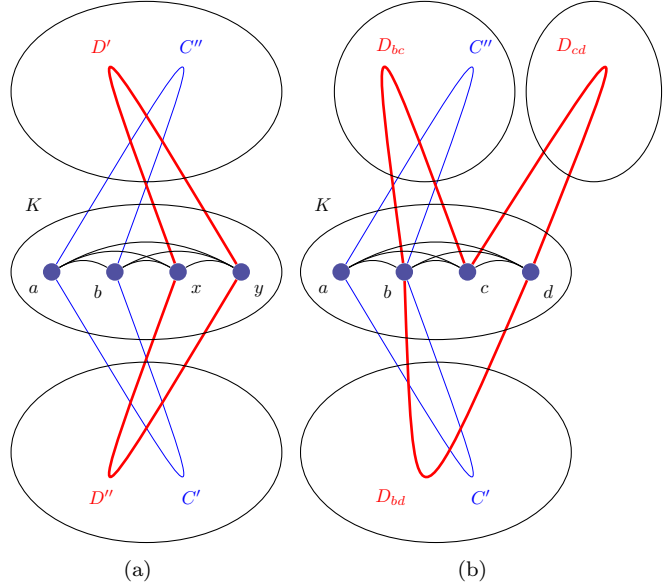
\begin{figure}[h]
        \centering
        \subfigure[ ]{
       \begin{tikzpicture}[scale=0.6, rotate=90, every node/.style={scale=0.7}]
    
    \node at (1.5,7) {$K$};
    
    \node at (0,6.6)[myblue,label=below left:$a$,fill=myblue, circle](a) {};
    \node at (0,5.2)[myblue,label=below left:$b$,fill=myblue, circle](b) {};  
    
    \node at (0,3.8)[myblue,label=below 
    right:$x$,fill=myblue, circle](x) {};
    \node at (0,2.4)[myblue,label=below right:$y$,fill=myblue, circle](y) {};

    \draw [color=blue] [line width=1.0pt, style =ultra thin] (a) .. controls (6, 3) .. (b);
    \draw [color=blue] [line width=1.0pt, style =ultra thin] (a) .. controls (-6,3) .. (b);
	\node at (5,3.5) [color=blue]{$C''$};
    \node at (-5,3.5) [color=blue]{$C'$};    
    
    \draw [color=red] [line width=1.0pt] (x) .. controls (6, 6) .. (y);
    \draw [color=red][line width=1.0pt] (y) .. controls (-6, 6) .. (x);
	\node at (5,5.5) [color=red]{$D'$};
    \node at (-5,5.5) [color=red]{$D''$};    
    
     \draw (0,4.5) ellipse (1.5 and 3);
    \draw (4,4.5) ellipse (2 and 3);
    \draw (-4,4.5) ellipse (2 and 3);
    
    \path	
    
    (a) edge [bend left] (b)
    (a) edge [bend left] (x)
    (a) edge [bend left] (y)
    (b) edge [bend left] (x)
    (b) edge [bend left] (y)
    (x) edge [bend left] (y)
    ;

    \end{tikzpicture}}
    \subfigure[ ]{
      \begin{tikzpicture}[scale=0.6, rotate=90, every node/.style={scale=0.7}]
    
    \node at (1.5,7) {$K$};
    
    \node at (0,6.6)[myblue,label=below left:$a$,fill=myblue, circle](a) {};
    \node at (0,5.2)[myblue,label=below left:$b$,fill=myblue, circle](b) {};  
    
    \node at (0,3.8)[myblue,label=below 
    right:$c$,fill=myblue, circle](c) {};
    \node at (0,2.4)[myblue,label=below right:$d$,fill=myblue, circle](d) {};
      
    \draw [color=blue] [line width=1.0pt, style =ultra thin] (a) .. controls (6, 3) .. (b);
    \draw [color=blue] [line width=1.0pt, style =ultra thin] (a) .. controls (-6,3) .. (b);
	\node at (5,3.5) [color=blue]{$C''$};
    \node at (-5,3.5) [color=blue]{$C'$};    
    
       \draw [color=red] [line width=1.0pt] (b) .. controls (6, 6) .. (c);
    \draw [color=red] [line width=1.0pt] (c) .. controls (6, 0) .. (d);
    \draw [color=red][line width=1.0pt] (b) .. controls (-6, 5) .. (d);
	\node at (5,5.5) [color=red]{$D_{bc}$};
	\node at (5,1.5) [color=red]{$D_{cd}$};
    \node at (-5,5.5) [color=red]{$D_{bd}$};  
    
    \draw (0,4.5) ellipse (1.5 and 3);
    \draw (4,4.75) ellipse (2 and 2);
    \draw (4,1) ellipse (2 and 1.5);
    \draw (-4,4.5) ellipse (2 and 3);
    
    \path	
    
    (a) edge [bend left] (b)
    (a) edge [bend left] (d)
    (a) edge [bend left] (c)
    (b) edge [bend left] (d)
    (b) edge [bend left] (c)
    (c) edge [bend left] (d)
    ;
    \end{tikzpicture}}
          \caption{Situations in the proof of Lemma~\ref{lemma:main-lemma}.}
\label{fig:lemma-lct-2-crosses-nuevo-total}
\end{figure}

\end{case}
  
  \begin{case} Every longest cycle intersects $K$ at least three times.
  
  As~${\lct(G) > \max \{1, \omega(G)-3\}} \geq k-3$,
  for every triangle~${\Delta \subseteq K}$, there exists a longest cycle that 3-intersects $K$ at $\Delta$.
 Suppose by contradiction that none of these cycles is a 3-attractor for~$K$.
Then, for every triangle~${\Delta \subseteq K}$, there exists
a longest cycle~$C_{\Delta}$ that 3-crosses~$K$ at~$\Delta$.
Let ${\Delta \subseteq K}$.
As $C_{\Delta}$ 3-crosses~$K$ 
at~$\Delta$,~$\Delta$ has at least two~$C_{\Delta}$-breaking vertices.
As there are~${|K| \choose 3}$ triangles in~$K$,
by pigeonhole principle, there exists
a vertex~$x \in K$
such that~$x$ is a breaking vertex for at least
$\frac{(|K|-1)(|K|-2)}{3}$
of the triangles incident to~$x$.

Suppose for a moment that $|K|\geq 5$.
Then, there exists two edge-disjoint triangles incident to~$x$
such that~$x$ is a breaking vertex for both of them.
Let~$xab$ and~$xcd$ be such triangles,
and let~$C$ and~$D$ be the corresponding cycles respectively.
As~$x$ breaks both~$C$ and~$D$,
without loss of generality we may assume that
${C_{xa} \nsim_{K} D_{xc}}$
and that~$C_{xb} \nsim_{K} D_{xd}$.
Also, there exists a part~${P \in \{D_{xc},D_{xd}\}}$ such that~${C_{ab} \nsim_{K} P}$ and
a part~${Q \in \{C_{xa},C_{xb}\}}$ such that~${D_{cd} \nsim_{K} Q}$.
Without loss of generality, we may suppose that~$C_{ab} \nsim_{K} D_{xd}$.
If~$D_{cd} \nsim_{K} C_{xa}$ then
$D_{xc} \cdot D_{cd} \cdot da \cdot C_{ax}$
and~$D_{dx} \cdot C_{xb} \cdot C_{ba} \cdot ad$ are cycles,
a contradiction (Figure~\ref{fig:lct-3-crosses}(a)).
So~$D_{cd} \sim_{K} C_{xa}$, implying that~$D_{cd} \nsim_{K} C_{xb}$.
If~$C_{ab} \nsim_{K} D_{cd}$ then
$C_{xa} \cdot ac \cdot D_{cx}$
and~$C_{xb} \cdot C_{ba} \cdot ac \cdot D_{cd} \cdot D_{dx}$
are cycles, a contradiction (Figure~\ref{fig:lct-3-crosses}(b)).
So,~$C_{ab} \sim_{K} D_{cd}$.
As~$D_{cd} \sim_{K} C_{xa}$, we conclude
that~$C_{ab} \sim_{K} C_{xa}$.
As~$C_{xa} \nsim_{K} D_{xc}$,
we conclude that~$C_{ab} \nsim_{K} D_{xc}$.
Then,~$C_{xa} \cdot C_{ab} \cdot bc \cdot D_{cx}$
and~$C_{xb} \cdot bc \cdot D_{cd} \cdot D_{dx}$ are both cycles, again a contradiction (Figure~\ref{fig:lct-3-crosses}(c)).

\begin{figure}[ht]
\centering

\subfigure[ ]{
\begin{tikzpicture}[thick, scale=0.75,
  every node/.style={draw,circle},
  psnode/.style={fill=myblue},
  qsnode/.style={fill=mygreen},
  rsnode/.style={fill=myred},
  style1/.style={ellipse,draw,inner sep=-1pt,text width=1.5cm},
  style2/.style={ellipse,draw,inner sep=0pt,text width=1cm,text heigth=16cm,ellipse ratio=2},
]
\begin{scope}[start chain=going below,node distance=7mm,yshift=1.5cm]
\node[psnode,on chain] (Cxa) [label=left: $C_{xa}$] {};
\node[psnode,on chain] (Cxb) [label=left: $C_{xb}$] {};
\node[psnode,on chain] (Cab) [label=left: $C_{ab}$] {};
\end{scope}
\begin{scope}[xshift=2cm,yshift=1.5cm,start chain=going below,node distance=7mm]
\node[qsnode,on chain] (Dxc) [label=right: $D_{xc}$] {};
\node[qsnode,on chain] (Dxd) [label=right: $D_{xd}$] {};
\node[qsnode,on chain] (Dcd) [label=right: $D_{cd}$] {};
\end{scope}
\draw (Cxa) -- (Dxc);
\draw (Cxb) -- (Dxd);	
\draw (Cab) -- (Dxd);
\draw (Cxa) -- (Dcd);
\end{tikzpicture}
}\subfigure[ ]{
\begin{tikzpicture}[thick,scale=0.75,
  every node/.style={draw,circle},
  psnode/.style={fill=myblue},
  qsnode/.style={fill=mygreen},
  rsnode/.style={fill=myred},
  style1/.style={ellipse,draw,inner sep=-1pt,text width=1.5cm},
  style2/.style={ellipse,draw,inner sep=0pt,text width=4cm,text heigth=16cm,ellipse ratio=2},
]
\begin{scope}[start chain=going below,node distance=7mm,yshift=1.5cm]
\node[psnode,on chain] (Cxa) [label=left: $C_{xa}$] {};
\node[psnode,on chain] (Cxb) [label=left: $C_{xb}$] {};
\node[psnode,on chain] (Cab) [label=left: $C_{ab}$] {};
\end{scope}
\begin{scope}[xshift=2cm,yshift=1.5cm,start chain=going below,node distance=7mm]
\node[qsnode,on chain] (Dxc) [label=right: $D_{xc}$] {};
\node[qsnode,on chain] (Dxd) [label=right: $D_{xd}$] {};
\node[qsnode,on chain] (Dbd) [label=right: $D_{cd}$] {};
\end{scope}
\draw (Cxa) -- (Dxc);
\draw (Cxb) -- (Dxd);
\draw (Cxb) -- (Dcd);	
\draw (Cab) -- (Dcd);
\draw (Cab) -- (Dxd);
\end{tikzpicture}
}\subfigure[ ]{
\begin{tikzpicture}[thick,scale=0.75,
  every node/.style={draw,circle},
  psnode/.style={fill=myblue},
  qsnode/.style={fill=mygreen},
  rsnode/.style={fill=myred},
  style1/.style={ellipse,draw,inner sep=-1pt,text width=1.5cm},
  style2/.style={ellipse,draw,inner sep=0pt,text width=4cm,text heigth=16cm,ellipse ratio=2},
]
\begin{scope}[start chain=going below,node distance=7mm,yshift=1.5cm]
\node[psnode,on chain] (Cva) [label=left: $C_{xa}$] {};
\node[psnode,on chain] (Cvb) [label=left: $C_{xb}$] {};
\node[psnode,on chain] (Cab) [label=left: $C_{ab}$] {};
\end{scope}
\begin{scope}[xshift=2cm,yshift=1.5cm,start chain=going below,node distance=7mm]
\node[qsnode,on chain] (Dvc) [label=right: $D_{xc}$] {};
\node[qsnode,on chain] (Dvd) [label=right: $D_{xd}$] {};
\node[qsnode,on chain] (Dbd) [label=right: $D_{cd}$] {};
\end{scope}
\draw (Cxa) -- (Dxc);
\draw (Cxb) -- (Dxd);
\draw (Cxb) -- (Dcd);
\draw (Cab) -- (Dxc);		
\end{tikzpicture}}
\subfigure[ ]{
\begin{tikzpicture}[thick, scale=0.75,
  every node/.style={draw,circle},
  psnode/.style={fill=myblue},
  qsnode/.style={fill=mygreen},
  rsnode/.style={fill=myred},
  style1/.style={ellipse,draw,inner sep=-1pt,text width=1.5cm},
  style2/.style={ellipse,draw,inner sep=0pt,text width=4cm,text heigth=16cm,ellipse ratio=2},
]
\begin{scope}[start chain=going below,node distance=7mm,yshift=1.5cm]
\node[psnode,on chain] (Cab) [label=left: $C_{xb}$] {};
\node[psnode,on chain] (Cbd) [label=left: $C_{bd}$] {};
\node[psnode,on chain] (Cad) [label=left: $C_{xd}$] {};
\end{scope}

\begin{scope}[xshift=2cm,yshift=1.5cm,start chain=going below,node distance=7mm]
\node[qsnode,on chain] (Dac) [label=right: $D_{xc}$] {};
\node[qsnode,on chain] (Dcd) [label=right: $D_{cd}$] {};
\node[qsnode,on chain] (Dad) [label=right: $D_{xd}$] {};
\end{scope}

\draw (Cab) -- (Dcd);
\draw (Cbd) -- (Dac);	
\draw (Cab) -- (Dad);	
\draw (Cad) -- (Dac);	
\end{tikzpicture}
}\subfigure[ ]{
\begin{tikzpicture}[thick, scale=0.75,
  every node/.style={draw,circle},
  psnode/.style={fill=myblue},
  qsnode/.style={fill=mygreen},
  rsnode/.style={fill=myred},
  style1/.style={ellipse,draw,inner sep=-1pt,text width=1.5cm},
  style2/.style={ellipse,draw,inner sep=0pt,text width=4cm,text heigth=16cm,ellipse ratio=2},
]

\begin{scope}[start chain=going below,node distance=7mm,yshift=1.5cm]
\node[psnode,on chain] (Cab) [label=left: $C_{xb}$] {};
\node[psnode,on chain] (Cbd) [label=left: $C_{bd}$] {};
\node[psnode,on chain] (Cad) [label=left: $C_{xd}$] {};
\end{scope}

\begin{scope}[xshift=2cm,yshift=1.5cm,start chain=going below,node distance=7mm]
\node[qsnode,on chain] (Dac) [label=right: $D_{xc}$] {};
\node[qsnode,on chain] (Dcd) [label=right: $D_{cd}$] {};
\node[qsnode,on chain] (Dad) [label=right: $D_{xd}$] {};
\end{scope}

\draw (Cab) -- (Dcd);
\draw (Cbd) -- (Dac);	
\draw (Cab) -- (Dac);	
\draw (Cad) -- (Dad);	
\draw (Cbd) -- (Dcd);	
\end{tikzpicture}

}\subfigure[ ]{
\begin{tikzpicture}[thick, scale=0.75,
  every node/.style={draw,circle},
  psnode/.style={fill=myblue},
  qsnode/.style={fill=mygreen},
  rsnode/.style={fill=myred},
  style1/.style={ellipse,draw,inner sep=-1pt,text width=1.5cm},
  style2/.style={ellipse,draw,inner sep=0pt,text width=4cm,text heigth=16cm,ellipse ratio=2},
]
\begin{scope}[start chain=going below,node distance=7mm,yshift=1.5cm]
\node[psnode,on chain] (Cab) [label=left: $C_{xb}$] {};
\node[psnode,on chain] (Cbd) [label=left: $C_{bd}$] {};
\node[psnode,on chain] (Cad) [label=left: $C_{xd}$] {};
\end{scope}

\begin{scope}[xshift=2cm,yshift=1.5cm,start chain=going below,node distance=7mm]
\node[qsnode,on chain] (Dac) [label=right: $D_{xc}$] {};
\node[qsnode,on chain] (Dcd) [label=right: $D_{cd}$] {};
\node[qsnode,on chain] (Dad) [label=right: $D_{xd}$] {};
\end{scope}

\draw (Cab) -- (Dac);
\draw (Cab) -- (Dcd);
\draw (Cbd) -- (Dac);	
\draw (Cad) -- (Dad);	
\draw (Cbd) -- (Dad);
\draw (Cbd) -- (Dad);
\draw (Cad) -- (Dcd);

\end{tikzpicture}
}
\caption{
Each bipartite graph represents the situation of the cycles~$C$ and~$D$ in
the proof of Lemma~\ref{lemma:main-lemma}. 
Each side of the bipartition 
has three vertices that represent the parts of each cycle.  
There is a straight edge in the graph 
if the corresponding parts, say~$P$ and~$Q$, are such that~$P \nsim_{K} Q$.}
\label{fig:lct-3-crosses}
\end{figure}
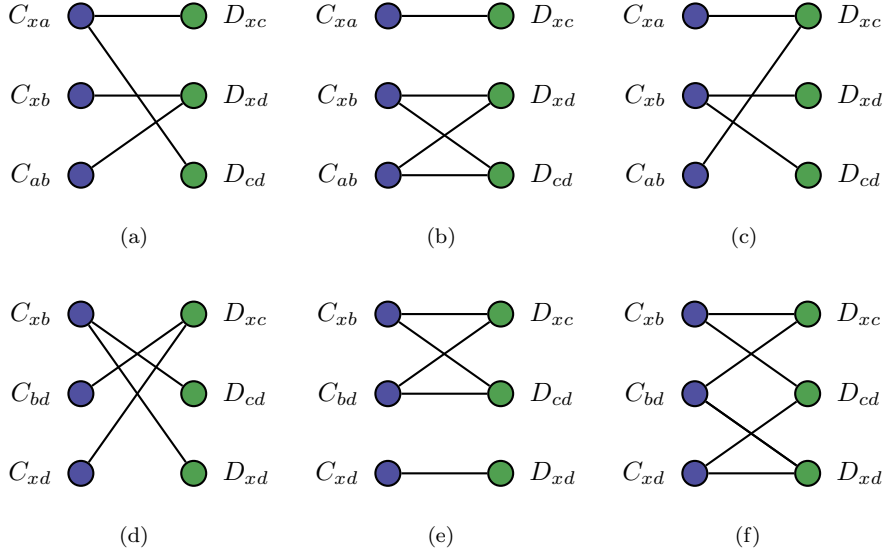

Now suppose that $|K|=4$.
Then $x$ is a breaking vertex for two triangles incident to $x$.
Let~$xbd$ and~$xcd$ be these two triangles.
Let~$C$ and~$D$ be the corresponding longest cycles respectively.
Hence,
\begin{equation}\label{eq:C_ab_nsim_C_ad_and_D_ac_nsim_D_ad}
C_{xb} \nsim_{K} C_{xd} \mbox{ and } D_{xc} \nsim_{K} D_{xd}.
\end{equation}
Also, 
note that
\begin{equation}\label{eq:C_ab_nsim_D_cd_and_C_bd_nsim_D_ac}
C_{xb} \nsim_{K} D_{cd} \mbox{ and } C_{bd} \nsim_{K} D_{xc}.
\end{equation}
Is the proof of \eqref{eq:C_ab_nsim_D_cd_and_C_bd_nsim_D_ac}:
by Proposition \ref{prop:core-lpt-lct-Branch_t(P)=Branch_t(t')},
there exists
an edge $tt' \in E(T)$
such that~${\Branch_{t}(C_{xb})=\Branch_{t}(t')}$.
As both $V_t$ and $V_{t'}$
are maximal cliques in $G$, there exists a vertex in $V_t \setminus V_{t'} \subseteq \{c,d\}$.
Thus,
by Proposition \ref{prop:core-sep-tt'}, $V_t \cap V_{t'}$ separates~$C_{xb}$ from ${D_{cd}}$;
hence ${C_{xb} \nsim D_{cd}}$.
Analogously,
$C_{bd} \nsim D_{xc}$.

By \eqref{eq:C_ab_nsim_C_ad_and_D_ac_nsim_D_ad},
either
$C_{xb} \nsim D_{xd}$
and~$C_{xd} \nsim D_{xc}$,
or
$C_{xb} \nsim D_{xc}$ and~$C_{xd} \nsim D_{xd}$.
In the first case, by \eqref{eq:C_ab_nsim_D_cd_and_C_bd_nsim_D_ac},
$C_{bx} \cdot D_{xd} \cdot D_{dc} \cdot cb$
and 
$D_{xc} \cdot cb \cdot C_{bd} \cdot C_{dx}$ are cycles, one of them longer than~$L$,
a contradiction (Figure~\ref{fig:lct-3-crosses}(d)).

In the second case,
${C_{bd} \sim D_{cd}}$.
Indeed, suppose for a moment that
${C_{bd} \nsim D_{cd}}$.
Thus,~${C_{xd} \cdot D_{xd}}$
and~${C_{xb}\cdot C_{bd} \cdot D_{dc} \cdot D_{cx}}$
are cycles (Figure~\ref{fig:lct-3-crosses}(e)).
But then,~${C_{xd} \cdot D_{xd}}$ is a longest cycle that
2-intersects~$V_t$, a contradiction.
Hence ${C_{bd} \sim D_{cd}}$.
If $C_{bd} \nsim D_{xd}$ and $C_{xd} \nsim D_{cd}$, then
$C_{bx}\cdot C_{xd} \cdot dc \cdot D_{cd}$ and
$D_{cx}\cdot D_{xd} \cdot C_{db} \cdot bc$
are cycles, a contradiction (Figure~\ref{fig:lct-3-crosses}(f)).
So we may assume, without loss of generality, that $C_{xd} \sim D_{cd}$.
Thus, we have, $C_{bd} \sim C_{xd} \sim D_{cd}$. By Proposition \ref{prop:core-sep-tt'}, there
exists an edge $tt' \in E(G)$ such that
$\Branch_t(t') = \Branch_t(C_{xd})$.
As both $V_t$ and $V_{t'}$ are maximal
cliques in $G$, there exists a vertex
in $V_t \setminus V_{t'} \subseteq \{b,c\}$.
If $V_t \setminus V_{t'} = \{b\}$, then by
Proposition \ref{prop:core-sep-tt'}, $V_t \cap V_{t'}$ separates $C_{xd}$ from $C_{bd}$, a contradiction.
If $V_t \setminus V_{t'} = \{c\}$, then by
Proposition~\ref{prop:core-sep-tt'}, $V_t \cap V_{t'}$ separates $C_{xd}$ from $C_{cd}$, again a contradiction.
\end{case}
This concludes the proof of the lemma.
\end{proof}

%% file: defs-fenced-crosses-3.tex
\begin{tikzpicture}
[scale=0.7, label distance=3pt, every node/.style={fill,circle,inner sep=0pt,minimum size=6pt}]
 
    \node at (0,0)[myblue,label=below left :$c$,fill=black, circle](c) {};
    \node at (3,0)[myblue,label=below right:$d$,fill=black, circle](d) {};
    \node at (1.5,2.5981)[myblue,label=above :$b$,fill=black, circle](b) {};
    \node at (3.5,2.5981)[myblue,label=above right:$v_1$,fill=black, circle](v1) {};
    \node at (4.25,1.4)[myblue,label=below right:$v_2$,fill=black, circle](v2) {};
    \node at (-0.5,2.5981)[myblue,label=above left:$v_3$,fill=black, circle](v3) {};

    \node at (1.5,0.866)[myblue,label=above   left:$a$,fill=black, circle](a) {};
    \node at (-1.25,1.4)[myblue,label=below left:$v_4$,fill=black, circle](v4) {};
    \node at (1.5,-1.4)[myblue,label=below left:$v_5$,fill=black, circle](v5) {};

    \draw  (a) -- (b);
    \draw  [dashed] [color=red] [line width=1.0pt] (b) -- (c);
    \draw (a) -- (c);
    \draw  (a) -- (d);
    \draw  (b) -- (d);
    \draw (c) -- (d);
    
    \draw (b) -- (v1);
    \draw (d) -- (v1);
    \draw (d) -- (v2);
    \draw  [dashed] [color=red] [line width=1.0pt] (b) -- (v2);
    \draw  [dashed] [color=red] [line width=1.0pt] (a) -- (v2);
    
    \draw  (b) -- (v3);
    \draw   (c) -- (v4);
    \draw (v3) -- (v4);
    \draw (v3) -- (c);
    \draw (v4) -- (b);
    
    \draw (a) -- (v1);
    \draw  [dashed] [color=red] [line width=1.0pt] (a) -- (v5);
    \draw  [dashed] [color=red] [line width=1.0pt] (v5) -- (c);

\end{tikzpicture}

%% file: sectionConcludingRemarks.tex
In this paper, we showed upper bounds for the minimum cardinality of a set
of vertices
that intersects all longest cycles in a 2-connected partial $k$-tree and in a 2-connected chordal graph.
We showed that, in partial $k$-trees, there is a set of at most
$k-1$ vertices that intersects all longest cycles of the graph,
and that in chordal graphs there is such a set with cardinality at most $\max\{1,\omega-3\}$, where $\omega$ is the cardinality of a maximum clique of the graph.
This implies that all longest cycles intersect
in partial 2-trees and in 3-trees.

The question of whether~$\lct(G)=1$ 
when $G$ is a 2-connected chordal graph
is still open, we conjecture a positive answer to that question.
As any graph is a partial $k$-tree for some~$k$, we have that $\lct(G)>1$ when~$G$ is a 2-connected partial $k$-tree.
However, for partial 3-trees, it has been proved that all longest cycles intersect \cite{Gutierrez18, JuanPhd}.
For partial 4-trees, there exists a 2-connected graph~$G$
given by Thomassen on
15 vertices~\cite[Figure 16]{Shabbir13}, with $\tw(G)=4$ and~${\lct(G)=2}$.
Hence,
by Theorem~\ref{theorem:lct<tw-1},
we conclude the following corollary and conjecture
that ${\ell = 2}$.
\begin{corollary}\label{corollary:lct-partial-4-tree}
Let~$\ell$ be the minimum integer such that
$\lct(G) \leq \ell$ for every 2-connected partial 4-tree~$G$.
Then,~$\ell \in \{2,3\}$.
\end{corollary}
Transversals of longest paths has been also studied
\cite{Cerioli2020,Cerioli2019,Chen17,deRezende13,Golan2018,Joos15, RautenbachS14}.
Also,
other questions about intersection of longest cycles have been rised
by several authors
\cite{Chen98,Hippchen08,Jendrol97,Stewart95}.